\documentclass[preprintnumbers,11pt,onecolumn]{article}
\pdfoutput=1	

\usepackage{fullpage}
\usepackage{amsfonts, amssymb, amsmath, amsthm}
\usepackage{latexsym}
\usepackage[tracking=smallcaps]{microtype}	
\usepackage{url}
\usepackage{color}
\definecolor{DarkGray}{rgb}{0.1,0.1,0.5}
\usepackage[colorlinks=true,breaklinks, linkcolor=black,citecolor=black,urlcolor=black]{hyperref}	

\usepackage{graphicx}
\usepackage[tight, TABBOTCAP]{subfigure}

\newcommand{\ket}[1]{{|#1\rangle}}

\newcommand{\abs}[1]{{\lvert #1\rvert}}	

\def\A {{\mathcal A}}
\def\B {{\mathcal B}}

\def\S {{\mathcal S}}

\newcounter{sprows}

\newlength{\spheight}
\newlength{\spraise}

\newcommand{\comment}[1]{\emph{\color{blue}Comment:\color{black} #1}} 
\newlength{\commentslength}
\newcommand{\comments}[1]{
\hspace{-2\parindent}
\addtolength{\commentslength}{-\commentslength}
\addtolength{\commentslength}{\linewidth}
\addtolength{\commentslength}{-\parindent}
\fcolorbox{blue}{white}{\smallskip\begin{minipage}[c]{\commentslength}
\emph{Comments:}\begin{itemize}#1\end{itemize}\end{minipage}}\bigskip
}
\newcommand{\rem}[1]{}

\newtheorem{theorem}{Theorem}[section]
\newtheorem{lemma}[theorem]{Lemma}

\newtheorem{claim}[theorem]{Claim}

\newtheorem{proposition}[theorem]{Proposition}

\newtheorem{definition}[theorem]{Definition}

\newtheorem{example}[theorem]{Example}

\newfont{\subsubsecfnt}{ptmri8t at 11pt}
\renewcommand{\subparagraph}[1]{\smallskip{\subsubsecfnt #1.}}

\newcommand{\eqnref}[1]{\hyperref[#1]{{(\ref*{#1})}}}
\newcommand{\thmref}[1]{\hyperref[#1]{{Theorem~\ref*{#1}}}}
\newcommand{\lemref}[1]{\hyperref[#1]{{Lemma~\ref*{#1}}}}
\newcommand{\corref}[1]{\hyperref[#1]{{Corollary~\ref*{#1}}}}
\newcommand{\defref}[1]{\hyperref[#1]{{Definition~\ref*{#1}}}}
\newcommand{\secref}[1]{\hyperref[#1]{{Section~\ref*{#1}}}}
\newcommand{\figref}[1]{\hyperref[#1]{{Figure~\ref*{#1}}}}
\newcommand{\tabref}[1]{\hyperref[#1]{{Table~\ref*{#1}}}}
\newcommand{\remref}[1]{\hyperref[#1]{{Remark~\ref*{#1}}}}
\newcommand{\appref}[1]{\hyperref[#1]{{Appendix~\ref*{#1}}}}
\newcommand{\claimref}[1]{\hyperref[#1]{{Claim~\ref*{#1}}}}
\newcommand{\factref}[1]{\hyperref[#1]{{Fact~\ref*{#1}}}}
\newcommand{\propref}[1]{\hyperref[#1]{{Proposition~\ref*{#1}}}}
\newcommand{\exampleref}[1]{\hyperref[#1]{{Example~\ref*{#1}}}}
\newcommand{\conjref}[1]{\hyperref[#1]{{Conjecture~\ref*{#1}}}}

\allowdisplaybreaks[1]

\def\beq{\begin{equation}}
\def\eeq{\end{equation}}

\def\COLOR{}
\ifdefined\COLOR
\definecolor{alicered}{rgb}{1,0,0}
\definecolor{bobblue}{rgb}{0.196, 0.365, 0.745}
\definecolor{charliegreen}{rgb}{0, 0.522, 0.216}
\else
\definecolor{alicered}{rgb}{0,0,0}
\definecolor{bobblue}{rgb}{0,0,0}
\definecolor{charliegreen}{rgb}{0,0,0}
\fi

\newcommand{\Alice}{{\color{black} Alice}}
\newcommand{\Bob}{{\color{black} Bob}}
\newcommand{\Bobj}[1] {{\color{black} Bob$_{#1}$}}
\def \Bobone {\Bobj{1}}
\def \Bobtwo {\Bobj{2}}
\newcommand{\Charlie}{{\color{black} Charlie}}
\newcommand{\Charliej}[1] {{\color{black} Charlie$_{#1}$}}

\renewcommand{\A}{{\color{black}A}}
\renewcommand{\B}{{\color{black}B}}

\newcommand{\Bj}[1]{{\color{black}B_{#1}}}
\newcommand{\Cj}[1]{{\color{black}C_{#1}}}

\newcommand{\valueclassical}{\omega_c}
\newcommand{\valuequantum}{\omega_q}

\mathchardef\mhyphen="2D

\newcommand{\randlocal}[2]{{r_{\! \scriptscriptstyle {#1} \shortrightarrow {#2}}}}

\usepackage[shortlabels]{enumitem}  

\usepackage{array}
\usepackage{stmaryrd}  

\newcommand{\Q}{{\mathcal Q}}  
\renewcommand{\comment}[1]{}\renewcommand{\comments}[1]{}

\begin{document}
\def\compilefullpaper{}

\title{Test to separate quantum theory from non-signaling theories}
\author{Rui Chao \qquad Ben W. Reichardt \\ University of Southern California}
\date{}

\maketitle

\begin{abstract}
A Bell test separates quantum mechanics from a classical, local realist theory of physics.  However, a Bell test cannot separate quantum physics from all classical theories.  Classical devices supplemented with non-signaling correlations, e.g., the Popescu-Rohrlich ``nonlocal box," can pass a Bell test with probability at least as high as any quantum devices can.  After all, quantum entanglement does not allow for signaling faster than the speed of light, so in a sense is a weaker special case of non-signaling correlations.  It could be that underneath quantum mechanics is a deeper non-signaling theory.  

We present a test to separate quantum theory from powerful non-signaling theories.  The test extends the CHSH game to involve three space-like separated devices.  Quantum devices sharing a three-qubit GHZ state can pass the test with probability $5.1\%$ higher than classical devices sharing arbitrary non-signaling correlations between pairs.  

More generally, we give a test that $k$ space-like separated quantum devices can pass with higher probability than classical devices sharing arbitrary $(k-1)$-local non-signaling correlations.  
\end{abstract}

\section{Introduction}

Is quantum physics correct and complete, or is there a deeper physical theory underneath it?  
Einstein, Podolsky and Rosen~\cite{EinsteinPodolskyRosen35epr} proposed that quantum mechanics might lie above a deterministic, classical theory for physics.  This possibility can be tested.  Bell~\cite{Bell64epr} gave a test, refined by Clauser, Horne, Shimony and Holt~\cite{ClauserHorneShimonyHolt69chshgame}, that can be passed by quantum-mechanical systems, but not by a deterministic classical theory in which faster-than-light communication is impossible.  Recently, several groups have demonstrated ``loophole-free" Bell-inequality violations~\cite{Hensen15loopholefreeCHSH, Hensen16loopholetwo, Shalm15loopholeCHSH, Giustina15loopholeCHSH, Rosenfeld16loopholefree}, i.e., systems that unambiguously pass the CHSH test.  Up to high statistical confidence, this rules out the local-hidden-variable models suggested in~\cite{EinsteinPodolskyRosen35epr}, giving strong evidence that quantum mechanics is correct and complete.   

However, other classical models beyond the local-hidden-variable models could govern reality.  In particular, non-signaling correlations are a nondeterministic classical model constrained not to allow faster-than-light communication~\cite{Rastall85nonlocalbox, KhalfinTsirelson85nonlocal, PopescuRohrlich93nonlocal}.  The CHSH test cannot rule out a non-signaling theory of physics.  The ``nonlocal box" violates the Bell-CHSH inequality maximally, beyond what is possible in quantum physics.  In fact, any two-party correlations achievable quantumly can be achieved with non-signaling distributions.  Thus non-signaling theories are typically thought of as more general and more powerful than quantum physics.  For example, a cryptographic security proof based on a non-signaling security assumption is less conservative and therefore stronger than a proof that assumes the validity of quantum mechanics~\cite{BarrettHardyKent04diqkd, AcinGisinMasanes06diqkdnosignaling, Scarani06secrecyextraction, Masanes09diqkdnosignalingcomposablesecurity, MasanesPironioAcin10deviceindependent, HanggiRennerWolf10diqkd, BarrettColbeckKent12twodevicediqkd, MasanesRennerChristandlWinterBarrett06DIQKDnosignalingcomposablesecurity, HanggiRennerWolf09nosignalingprivacyamplificationimpossible}.  

\begin{figure}[b]
\centering
\includegraphics[scale=.125]{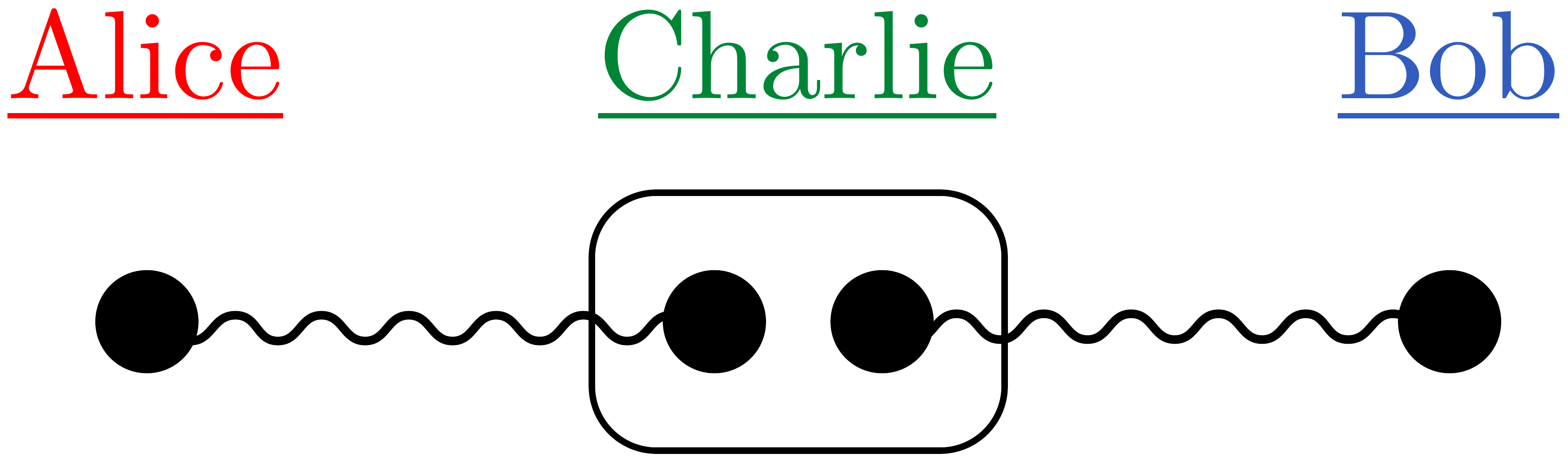}
\caption{In the Teleported CHSH game, \Alice\ and \Charlie, and \Bob\ and \Charlie\ share EPR states.  They cannot communicate.  \Charlie\ applies a Bell measurement, which might fail.  Conditioned on it succeeding, \Alice\ and \Bob\ have an EPR state with which they can win a CHSH game with probability~$85.4\%$.  If the initial EPR states are replaced with arbitrary \Alice-\Charlie\ and \Bob-\Charlie\ non-signaling resources, however, they can win with probability only $75\%$ (\thmref{t:postselectedteleport}).} \label{f:teleportation}
\end{figure}

We show here that, contrary to this intuition from the CHSH test, quantum entanglement can be more powerful than non-signaling resources.  We also provide a test that can separate quantum theory from a general class of non-signaling theories.  In particular, there is a three-party correlation that is achievable using quantum entanglement, but that cannot be achieved by classical parties using any combination of arbitrary two-party non-signaling resources (\thmref{t:distributedgame}).  In this aspect, quantum correlations are thus more powerful than two-local non-signaling correlations.  The separation is robust to constant error, and therefore it is in principle experimentally testable.  

\smallskip

The intuitive idea behind our results is that a non-signaling resource is a black box, that can only be accessed in a particular, classical way---whereas quantum correlations can be looked at in different directions and manipulated quantumly.  In particular, this means that unlike quantum entanglement, non-signaling resources cannot be teleported.  In fact, we show that CHSH games using entanglement-swapping to generate the quantum entanglement, as in~\cite{Hensen15loopholefreeCHSH}, are already testing a much simpler separation between two-local quantum and non-signaling resources (\figref{f:teleportation} and \thmref{t:postselectedteleport}).  Namely, the correlations they exhibit between three parties, \Alice, \Bob\ and \Charlie, cannot be achieved using two-party non-signaling resources between \Alice\ and \Charlie, and \Bob\ and \Charlie; a direct \Alice-\Bob\ non-signaling resource is required.  In contrast, quantumly, a Bell measurement on \Charlie's qubits of \Alice-\Charlie\ and \Bob-\Charlie\ EPR states can create entanglement between \Alice\ and \Bob.  

Our main theorem, \thmref{t:distributedgame}, gives a three-party game that quantum players can win with probability over $5.1\%$ higher than classical players sharing arbitrary two-local non-signaling correlations.  The game extends a CHSH game.  The three players, \Alice, \Bob\ and \Charlie, are meant to share a GHZ state $\tfrac{1}{\sqrt 2}(\ket{000} + \ket{111})$.  There are two sub-games.  \Charlie\ can measure in the $\sigma^x$ basis, leaving \Alice\ and \Bob\ with an EPR state (up to a possible $\sigma^z$ correction), with which they can play a CHSH game.  Alternatively, \Charlie\ can measure~$\sigma^z$, obtaining the same result as \Alice\ measuring~$\sigma^z$.  Importantly, a $\sigma^z$ measurement is part of the CHSH game, so \Alice\ cannot distinguish between the sub-games.  This makes it difficult for her to make use of any \Alice-\Bob\ or \Alice-\Charlie\ non-signaling resources.  (For example, dependence on the \Alice-\Bob\ resources should hurt her when playing the second sub-game with \Charlie.  This intuition can lead one astray, however, and the proof is somewhat subtle.  \appref{s:distributedgamecounterexample} defines games that are similar in spirit, but that do not separate quantum from two-local non-signaling theories.)  

Our test cannot rule out a three-party non-signaling correlation---after all, quantum theory is non-signaling.  Our test only rules out two-local non-signaling theories.  This is valuable because in a hypothetical model of classical physics supplemented by non-signaling resources, true three-local interactions would likely be more challenging to establish than two-local interactions.  Our test shows the impossibility of generating certain three-local non-signaling correlations from two-local correlations, showing a strong contrast between non-signaling resources and quantum entanglement.  

More generally, for any $k \geq 1$ we give a $(k+2)$-party game that quantum players sharing a GHZ state can win with strictly higher probability than classical devices sharing arbitrary $(k+1)$-local non-signaling correlations.  Lower bounds on the gaps are given in \figref{f:klocalgaps}.  Note that although the GHZ state is a $(k+2)$-local entangled state, it can be generated from two-local EPR states with an initial teleportation round, and this does not help the classical devices.  Thus, in a sense, two-local quantum correlations are more powerful than $(k+1)$-local classical non-signaling correlations.  

\smallskip

Barrett and Pironio~\cite{BarrettPironio05nonsignaling} have previously studied the problem of separating quantum correlations from two-local non-signaling correlations.  They give a five-player game, a ``graph game," that entangled quantum players can always win, but that classical players sharing arbitrary two-local non-signaling resources cannot win with probability one.  This result establishes that quantum theory can be more powerful than a classical theory with two-local non-signaling correlations.  
It has been extended to a $13$-player game that protects against four-local correlations, and $k^2 2^{2k-2}$-player games that protect against $k$-local correlations~\cite{AnshuMhalla12NSgraph}; and, non-constructively, to $n$ player games for sufficiently large~$n$ that protect against $< 0.11 n$ local correlations~\cite{HoyerMhallaPerdrix16graphgames}.  However, these results do not give any test to separate the theories.  The issue is that, potentially, the classical players could win with probability \emph{arbitrarily close} to one, so no experiment could statistically distinguish quantum from classical.  Such an eventuality would not be entirely surprising in light of nonlocality distillation: players sharing an unbounded number of noisy nonlocal boxes can use them in combination to implement a nonlocal box with an arbitrarily small positive noise rate~\cite{BrassardBuhrmanLindenMethodTappUnger05nonlocal}.  

\smallskip

Sections~\ref{s:nonsignaling} and~\ref{s:chsh} review the definitions of non-signaling correlations and the CHSH game.  \secref{s:impossibleteleport} analyzes the Teleported CHSH game of \figref{f:teleportation}.  \secref{s:extendedchshgame} defines a three-player Extended CHSH game, and in \secref{s:twolocalnsversusquantum} we upper bound the probability of winning the game using two-local non-signaling correlations.  \secref{s:openproblems} concludes with open problems.

\section{Non-signaling distributions} \label{s:nonsignaling}

\begin{definition}
A conditional probability distribution $\Pr[X, Y \vert A, B]$ is \emph{non-signaling} if 
\begin{align*}
&\text{$\Pr[X = x \vert A = a, B = b]$ does not depend on~$b$, and }  \\
&\text{$\Pr[Y = y \vert A = a, B = b]$ does not depend on~$a$.}
\end{align*}
\end{definition}

To interpret this definition, consider two parties, \Alice\ and \Bob.  \Alice's input is $a$ and her output a random variable~$X$, and \Bob's input and output are $b$ and~$Y$.  The non-signaling condition is a locality requirement, that the marginal distribution of \Alice's output should depend only on her input, and similarly for \Bob.  (A local-hidden-variable model, in contrast, adds a ``realism" constraint: $X$ should be a deterministic function of~$a$, and $Y$ a deterministic function of~$b$.)  Thus one cannot communicate from \Bob\ to \Alice\ by changing~$b$.  This conforms with relativity theory, in that information from \Bob\ should not be able to travel faster than the speed of light to \Alice.  A consequence of locality is that \Alice\ can choose her input and sample from~$X$ before \Bob\ has even decided on~$b$.  When later \Bob\ inputs~$b$, his output~$Y$ will depend on $a, b$ and~$X$.  

With more parties, the non-signaling condition is that no subset of the parties should be able to communicate to any other subset by changing their inputs~\cite{BarrettLindenMassarPironioPopescuRoberts04nonlocal}.  (Other references include~\cite{BarrettPironio05nonsignaling, Almeida10guessyourneighbor, PironioBancalScarani11tripartitenonlocal}.)  Mathematically, this is equivalent to the constraints 
\begin{equation*}
p(x_1, \ldots, \widehat{x_i}, \ldots, x_k \vert a_1, \ldots, a_i, \ldots a_k) 
= p(x_1, \ldots, \widehat{x_i}, \ldots, x_k \vert a_1, \ldots, a_i', \ldots a_k) 
\end{equation*}
for all $\vec x, \vec a$ and $a_i'$.  (That is, changing the $i$th party's input should not affect the marginal distribution for the other coordinates.)  A strictly weaker definition, that it should not be possible to communicate to any one party alone, has also been considered~\cite{BrandaoHarrow12definetti}.

\section{CHSH game} \label{s:chsh}

The CHSH game~\cite{ClauserHorneShimonyHolt69chshgame} involves two players, \Alice\ and \Bob.  Send independent, uniformly random bits $A$ and~$B$ to \Alice\ and \Bob, respectively.  The players reply with respective bits $X$ and~$Y$.  Accept if 
\begin{equation*}
X \oplus Y = A B
 \enspace .
\end{equation*}

The classical value of the game, i.e., the maximum probability with which classical players can win, using either a deterministic or randomized strategy, is 
\begin{equation*}
\valueclassical(\mathrm{CHSH}) = 3/4
 \enspace .
\end{equation*}

The quantum value, i.e., the maximum probability with which quantum players can win, using an arbitrary initial shared quantum state, is 
\begin{equation*}
\valuequantum(\mathrm{CHSH}) = \cos^2 \tfrac\pi8 \approx 85.4\%
 \enspace .
\end{equation*}
A strategy achieving this success probability uses a shared EPR state $\tfrac{1}{\sqrt 2}(\ket{00} + \ket{11})$.  On input~$0$ \Alice\ measures $\sigma^z = \big(\begin{smallmatrix}1&0\\0&-1\end{smallmatrix}\big)$, and on input~$1$ she measures $\sigma^x = \big(\begin{smallmatrix}0&1\\1&0\end{smallmatrix}\big)$.  On input~$b \in \{0,1\}$, \Bob\ measures $\tfrac{1}{\sqrt 2}(\sigma^z + (-1)^b \sigma^x)$.  Using this strategy, $\Pr[X \oplus Y = a b \vert A = a, B = b] = \cos^2 \tfrac\pi8$ for all $a, b$.  

With access to an appropriate non-signaling distribution, classical players can win with probability one.  The (Popescu-Rohrlich) nonlocal box~\cite{Rastall85nonlocalbox, KhalfinTsirelson85nonlocal, PopescuRohrlich93nonlocal} is a non-signaling distribution $\Pr[X=x, Y=y \vert A=a, B=b]$ in which $X \oplus Y = A B$ always.  

The CHSH game has the interesting property that a quantum or non-signaling strategy that beats the classical value must generate randomness, in certain technical senses~\cite{Pironioetal09randombell, ColbeckRennet11amplifyrandomness}.  We will use a rough contrapositive of this known property, namely that an arbitrary non-signaling strategy in which $\Pr[X = 0 \vert A = 0] = 1$ cannot beat the classical value: 

\begin{proposition} \label{t:chshdeterministica0}
In the CHSH game with \Alice's response to question $A = 0$ fixed to~$X = 0$, the non-signaling value is $3/4$, the same as the classical value.  
\end{proposition}

\noindent
The proof is given in \appref{s:chshnanalysis}, \propref{t:chshndeterministica0}.

\section{Non-signaling correlations cannot teleport} \label{s:impossibleteleport}

We begin by considering a three-party protocol in which correlations, non-signaling or quantum, are only allowed between two of the pairs of parties.  

\begin{definition}
The \emph{Teleported CHSH game} is a protocol with three players: \Alice, \Bob\ and \Charlie.  The players are not allowed to communicate with each other.  The verifier sends independent, uniformly random bits $A$ and $B$ to \Alice\ and \Bob, respectively.  She receives in return bits $X$ and $Y$ from the respective players, and two bits $Z_1$ and $Z_2$ from \Charlie.  The verifier accepts if 
\begin{equation} \label{e:teleportedCHSH}
X \oplus Y \oplus A Z_1 \oplus (1-A) Z_2 = A B
 \enspace .  
\end{equation}
\end{definition}

\figref{f:teleportation} illustrates a quantum strategy for the game.  Initially, \Alice\ and \Charlie\ share an EPR state, as do \Bob\ and \Charlie, but there is no entanglement between \Alice\ and \Bob.  \Charlie\ can apply a Bell measurement to his halves of the two shared EPR states to teleport an EPR state between \Alice\ and \Bob.  Then \Alice\ and \Bob\ play a CHSH game.  The teleportation corrections reported by \Charlie\ are known only to the verifier, and are used to adjust \Alice's reported measurement outcome.  (This accounts for the $A Z_1 \oplus (1-A) Z_2$ terms in Eq.~\eqnref{e:teleportedCHSH}.)  

\begin{theorem} \label{t:teleport}
In the Teleported CHSH game, 
\begin{enumerate}
\item 
If \Alice\ and \Charlie\ share an EPR state $\tfrac{1}{\sqrt 2}(\ket{00} + \ket{11})$, and \Bob\ and \Charlie\ share an EPR state, then the players can win with probability $\valuequantum(\mathrm{CHSH}) = \cos^2 \tfrac \pi 8 \approx 85.4\%$.  

\item 
If the players are classical and \Alice\ and \Charlie\ share arbitrary non-signaling resources, as do \Bob\ and \Charlie, but \Alice\ and \Bob\ do not share any nontrivial non-signaling resources, then the players can win with probability at most $\valueclassical(\mathrm{CHSH}) = 3/4$.  

\item 
If \Alice\ and \Bob\ share a nonlocal box, then they can win with probability one.  
\end{enumerate}
\end{theorem}

\begin{proof}
If \Alice\ and \Charlie\ sample a non-signaling distribution where \Charlie's input is fixed, then \Charlie's marginal output distribution is known and therefore can be sampled using shared randomness.  \Alice\ can sample from her conditional output distribution using her input.  Thus shared randomness allows for sampling from the non-signaling resource's outputs, where \Alice's output depends on her input but \Charlie's output does not.  

Repeating this argument, we see that a protocol using arbitrary non-signaling resources can be simulated using shared randomness provided that one party to each resource has no input.  
\end{proof}

A natural variant of the Teleported CHSH game allows for \Charlie's measurement to fail.  

\begin{theorem} \label{t:postselectedteleport}
Consider the three-party protocol in which the verifier sends independent, uniformly random bits $A$ and $B$ to \Alice\ and \Bob, respectively, receives bits $X$ and $Y$ back from the respective players, and either ``success" or ``failure" from \Charlie.  On ``success," the verifier accepts if $X \oplus Y = A B$.  

Then there is a quantum strategy, using \Alice-\Charlie\ and \Bob-\Charlie\ EPR states, for which $\Pr[\text{success}] = 1/4$ and $\Pr[\text{accept} \vert \text{success}] = \valuequantum(\mathrm{CHSH}) = \cos^2 \tfrac\pi8$.  

However, for any classical strategy with arbitrary \Alice-\Charlie, and \Bob-\Charlie\ non-signaling resources, but no \Alice-\Bob\ non-signaling resources, $\Pr[\text{accept} \vert \text{success}] \leq \valueclassical(\mathrm{CHSH}) = 3/4$.  
\end{theorem}

\thmref{t:postselectedteleport} is useful because this situation arises in experiments.  In the loophole-free Bell-inequality violation of Hensen et al.~\cite{Hensen15loopholefreeCHSH}, \Alice\ and \Bob's qubits are single electrons, which are each entangled with a photon.  The photons are sent to \Charlie, who attempts to project the photons onto a singlet state.  When successful, this ``entanglement-swapping" procedure entangles the electrons.  However, the success rate is only $6.4 \times 10^{-9}$, because of the difficulty of generating simultaneous photons and photon loss.  By \thmref{t:postselectedteleport}, \Charlie's low success rate does not create any loopholes; still $\Pr[\text{accept} \vert \text{success}] \leq 3/4$ in the non-signaling case with no \Alice-\Bob\ resources.  

Note that any two-player game~$G$ can be used in place of the CHSH game in \thmref{t:postselectedteleport}.  In the quantum case, \Charlie\ can attempt to teleport any needed entanglement between \Alice\ and \Bob, declaring success only for the trivial Bell measurement outcome(s).  The value of the game with \Alice-\Charlie\ and \Bob-\Charlie\ non-signaling resources is at most $\valueclassical(G)$.  

Unfortunately, in an experiment it is difficult to enforce that certain pairs of parties are allowed to share non-signaling resources, while other pairs are not.

\section{Extended CHSH game} \label{s:extendedchshgame}

We define the $k$-Extended CHSH game, or ``$\mathrm{CHSH} + k$" for short, with players \Alice, \Bob\ and $\text{\Charliej{1}}, \ldots, \text{\Charliej{k}}$.  With $A$ and $X$ we denote \Alice's input and output, respectively, and similarly $B$ and~$Y$ for \Bob, and $C_j$ and $Z_j$ for \Charliej{j}.  These messages are all single bits.  

There are two types of questions: 
\begin{itemize}
\item 
In a \emph{consistency} question, the inputs are $A = 0$ and $C_J = 0$ for a uniformly random index $J \in [k]$.  The verifier accepts if $X = Z_J$.  
\item 
In a \emph{game} question, the verifier sends \Alice\ $A \in \{0,1\}$ and \Bob\ $B \in \{0,1\}$.  If $A = 0$, the verifier accepts if $X = Y$.  If $A = 1$, then the verifier also sends $C_j = 1$ to \Charliej{j} for every~$j$, and she accepts if $(X \oplus Z_1 \oplus \cdots \oplus Z_k) \oplus Y = B$.  
\end{itemize}
A consistency question is chosen with probability $q = 1 - 2 / (3^k + 1)$.  (The precise value is set to optimize \claimref{t:markov} below.)  

$\mathrm{CHSH} + 0$ is a standard CHSH game.  
Observe that $\mathrm{CHSH} + k$ embeds into a simpler game, in which the verifier chooses inputs $A = C_1 = \cdots = C_k \in \{0,1\}$ and $B \in \{0,1\}$, and accepts if $X = Z_1 = \cdots = Z_k = Y$, when $A = 0$, or if $(X \oplus Z_1 \oplus \cdots \oplus Z_k) \oplus Y = B$, when $A = 1$.  In this latter game, the roles of \Alice\ and $\text{\Charliej{1}}, \ldots, \text{\Charliej{k}}$ are all symmetrical.  We have chosen a less symmetrical presentation to help organize our analysis.  

\begin{proposition} \label{t:distributedgamequantumvalue}
There exists a quantum strategy for the $\mathrm{CHSH} + k$ game, using the shared state $\tfrac{1}{\sqrt 2}(\ket{0^{k+2}} + \ket{1^{k+2}})$, such that 
\begin{equation*}\begin{aligned}
\Pr[\mathrm{win} \vert \text{\emph{consistency question}}] &= 1 \\
\Pr[\mathrm{win} \vert \text{\emph{game question}}] &= \cos^2 \tfrac\pi8 
 \enspace .
\end{aligned}\end{equation*}
Thus $\Pr[\mathrm{win}] = 1 - (1-q) \sin^2 \frac\pi8$.  
\end{proposition}

\begin{proof}
\Alice\ and $\text{\Charliej{1}}, \ldots, \text{\Charliej{k}}$ behave symmetrically: on input~$0$ each measures $\sigma^z = \big(\begin{smallmatrix}1&0\\0&-1\end{smallmatrix}\big)$, and on input~$1$ each measures $\sigma^x = \big(\begin{smallmatrix}0&1\\1&0\end{smallmatrix}\big)$.  On input~$b \in \{0,1\}$, \Bob\ measures $\tfrac{1}{\sqrt 2}(\sigma^z + (-1)^b \sigma^x)$.  These are the same measurements used by an optimal strategy for the standard CHSH game.  
\end{proof}

Note that on input $A = 0$, \Alice\ cannot distinguish between a game and a consistency question, nor between the different consistency questions.  Intuitively, for classical players, even if \Alice\ shares two-local non-signaling resources with the others, it should be difficult for her to use these correlations.  One should be very wary of this intuition, however, since simple variants of this construction, for which the intuition might seem equally valid, provably do not work.  See \appref{s:distributedgamecounterexample}.  

When $A = 0$, the verifier's decision to accept or reject depends only on the answers from \Alice\ and either \Bob\ or \Charliej{J}.  This property will be essential for our later analysis.

\section{Intuition and technical ideas for $\mathrm{CHSH} + 1$ separation} \label{s:twolocalnsversusquantum} 

\begin{theorem} \label{t:distributedgame}
In the Extended CHSH game, $\mathrm{CHSH} + 1$, classical players sharing two-local non-signaling resources can win with probability at most $7/8$.  
\end{theorem}

\noindent
This is $\frac{\sqrt 2 - 1}{8} > 5.17\%$ below the quantum winning probability lower bound from \propref{t:distributedgamequantumvalue}.  

Before going into detail, let us explain the intuition.  

Consider a consistency question, so $A = 0$.  The verifier's decision to accept or reject does not depend on \Bob.  Furthermore, given \Charlie's answer \Alice\ has only one correct response.  Therefore, if \Alice's answer depends very much on the randomness she gets from her correlations with \Bob, she will necessarily be wrong a substantial fraction of the time.  If the players win with high probability, then replacing her strategy on input $A = 0$ with one that is independent of the \Alice-\Bob\ correlations does not change the success probability by very much.  

Next consider a game question with $A = 0$.  Again, \Alice\ has a unique correct response given \Bob's answer, and the verifier's acceptance predicate does not depend on \Charlie.  (This is only true for $A = 0$.)  Intuitively, then, \Alice's modified strategy on input $A = 0$ cannot depend much on her correlations with \Charlie, either.  \Alice's response to $A = 0$ must be nearly deterministic, so \propref{t:chshdeterministica0} upper-bounds the players' success probability by the classical value~$\valueclassical(\mathrm{CHSH})$.  

In order to make the analysis rigorous, we need to explain what it means for \Alice's answer to depend on the randomness from her correlations with \Bob.\footnote{This concept is not obvious.  As an example, say that \Alice\ takes the output of a resource she shares with \Bob, feeds that as input to a resource shared with \Charlie, and then outputs its answer.  Consider, does this strategy ``depend on" the \Alice-\Bob\ resource?  How do we replace it with one that does not depend on that resource?}  
This begins with factoring the underlying randomness of a non-signaling distribution.  
Consider a two-party non-signaling distribution $p(x, y \vert a, b)$.  Since $p(x \vert a, b)$ is independent of~$b$, the distribution can be factored as $p(x, y \vert a, b) = p(x \vert a) p(y \vert a, b, x)$.  Without loss of generality, then, we may assume that the underlying sample space has the factorized form $\Omega = [0, 1] \times [0, 1]$; and for a uniformly random sample $(r, s) \in \Omega$, $r$ determines~$x$ from~$a$, and $s$ determines $y$ from $r, a, b$.  (That is, $x$ is a deterministic function of $a$ and~$r$, and $y$ is a deterministic function of $r, a, b$ and~$s$.)  

Call $p(x, y \vert a, b) = p(x \vert a) p(y \vert a, b, x)$ the \emph{left-factorization} of the resource, and $p(x, y \vert a, b) = p(y \vert b) p(x \vert a, b, y)$ the \emph{right-factorization}.  See \figref{f:factorizations}.  

\begin{example} \label{t:randomnessdependenceexample}
For example, consider a two-party resource that has one input bit $a$ from \Alice\ and no inputs from \Bob, and that outputs to \Alice\ and \Bob\ the same uniformly random bit $x$, independent of \Alice's input.  The randomness of this resource can be parameterized by a uniformly random $r \in [0, 1]$, 
\begin{align*}
a &= 0 \; \Longrightarrow \; x = \begin{cases} 0 & \text{if $r < 1/2$} \\ 1 & \text{if $r > 1/2$} \end{cases}&
a &= 1 \; \Longrightarrow \; x = \begin{cases} 1 & \text{if $r < 1/2$} \\ 0 & \text{if $r > 1/2$} \end{cases}
\end{align*}

An equivalent parameterization, with no dependence on~$a$, is 
\begin{equation*}
x = \begin{cases} 0 & \text{if $r < 1/2$} \\ 1 & \text{if $r > 1/2$} \end{cases}
\end{equation*}
\end{example}

\begin{figure}
\centering
\begin{tabular}{c@{$\qquad\qquad$}c@{$\qquad\qquad$}c}
\subfigure[]{\raisebox{.2cm}{\includegraphics[scale=.5]{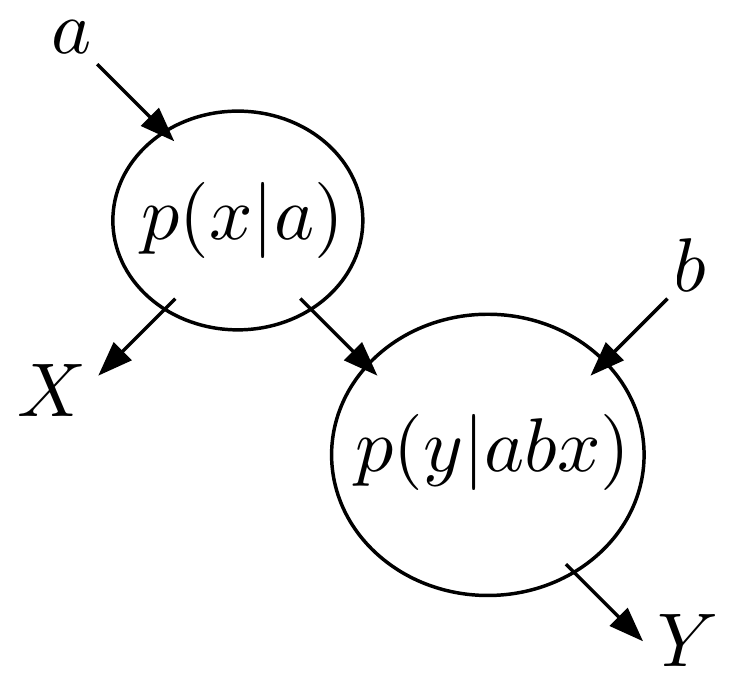}}} & 
\subfigure[]{\raisebox{.2cm}{\includegraphics[scale=.5]{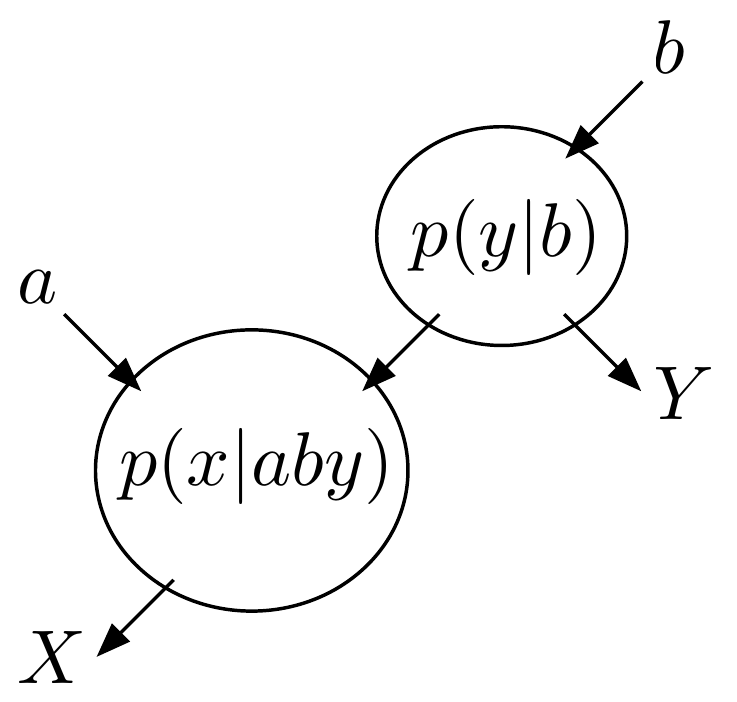}}} & 
\subfigure[\label{f:righttoleftfactorization}]{\raisebox{0cm}{\includegraphics[scale=.5]{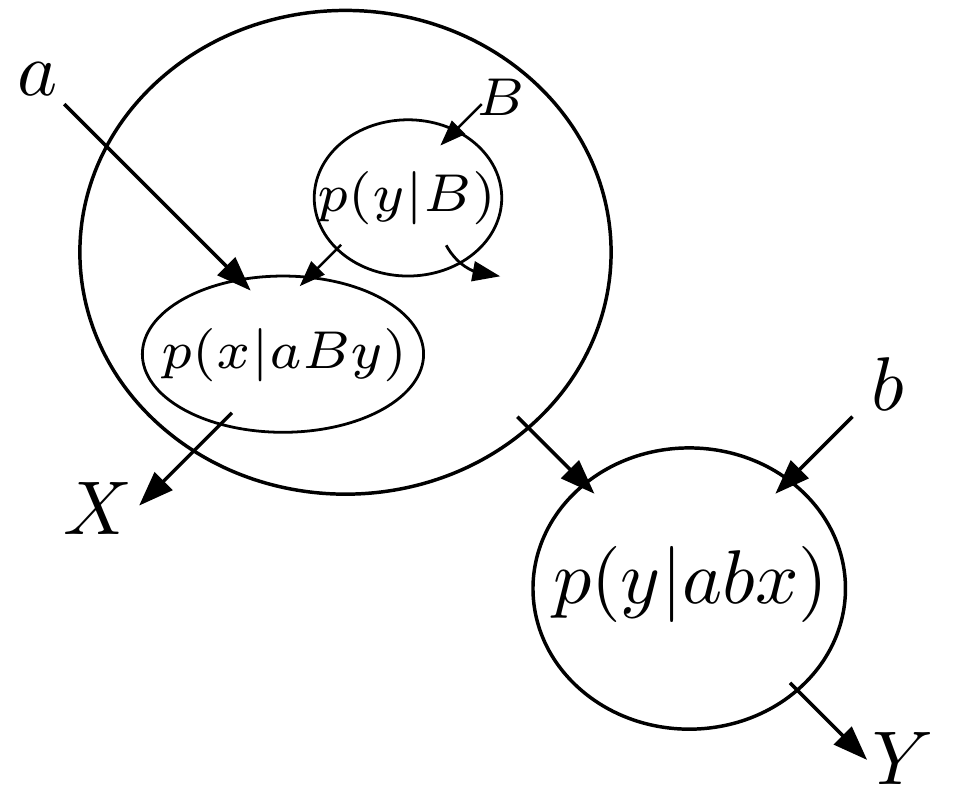}}}
\end{tabular}
\caption{
(a) A left-factorization of the non-signaling distribution $p(x, y \vert a, b)$ corresponds to \Alice\ asking her question first, and a right-factorization (b) to \Bob\ asking first.  
(c) If the internal randomness of a non-signaling resource is parameterized according to a right-factorization, this can be converted to a left-factorization by sampling from the distribution of \Bob's inputs and local randomness to determine \Alice's answer.  
} \label{f:factorizations}
\end{figure}

In the first parameterization in this example, for a fixed value of~$r$ changing \Alice's input~$a$ will change the output~$x$.  This would seem to be problematic for the preceding argument; \Alice\ could share the resource of \exampleref{t:randomnessdependenceexample} with \Charlie, and give it an input from an \Alice-\Bob\ resource.  Therefore even when $J = 2$, \Alice's outputs \emph{can} depend on the randomness from her correlations with \Bob.  This dependence is of course artificial; in the second parameterization of the example, $x$ has no dependence on~$a$.  

We will argue, then, that although \Alice's answer can depend on her correlations with \Bob, the resources can be reparameterized so that there is no such dependence.  This technical trick allows our intuitive argument to be pushed through for the case $J = \mathrm{\Charlie}$.  First, use a right-factorization for the interactions of \Alice\ and \Charlie, and a left-factorization for the \Alice-\Bob\ resources.  That is, parameterize the \Alice-\Charlie\ interactions by $s$, the input and local randomness to \Charlie.  \Charlie's outputs are a deterministic function of~$s$.  Therefore, given $s$ and \Alice's input, there is a unique valid answer for \Alice, so her answer can have little dependence on her interactions with \Bob.  Then, we switch to a left-factorization for the \Alice-\Charlie\ interactions, in which \Alice's local randomness (with respect to which her outputs are a deterministic function of her inputs) comes from \emph{sampling} $s$.  Roughly speaking, the resources on \Alice's side ``guess" \Charlie's input and randomness, and use the guess to determine \Alice's random outputs.  (It does not matter that the guess is almost certainly wrong; what matters is that this procedure generates the correct marginal distribution for \Alice's input/output transcripts.)  See \figref{f:righttoleftfactorization}.  In this left-factorization, it is still the case that \Alice's answer has little dependence on her interactions with \Bob.  

A technical problem is that \Alice\ potentially shares many non-signaling resources with \Charlie.  We want to use a left-factorization for all of them.  In the above reparameterization, based on sampling~$s$, all the resources must use the same~$s$.  They cannot make independent guesses.  This however correlates the randomness in the non-signaling resources, which the definition does not allow.\footnote{It is important to sample~$s$ at random; using a fixed value~$s^*$ would mean that \Alice's final output is constant with high probability, i.e., this would change the marginal distribution over \Alice's transcripts.}  For this reason we introduce \emph{multi-round} non-signaling resources.  All \Alice-\Charlie\ resources can be collected together into one multi-round resource, whose left-factorization samples~$s$ as above.

\section{General strategy simplification lemma}

As our tools are applicable beyond the setting of \thmref{t:distributedgame}, at this point it is appropriate to generalize.  In this section we will state and prove our main technical lemma, and in \secref{s:robustseparation} below we will apply it to Extended CHSH and Extended CHSH$_n$ games.  

For a vector $\vec v = (v_1, \ldots, v_m)$, we denote its first $j$ components by $\vec v_{1:j} = (v_1, \ldots, v_j)$, and we use $\vec v_{1:0}$ to remove the vector entirely from an expression.  

\begin{definition}
A $k$-party, multi-round non-signaling resource is a conditional probability distribution $\Pr[\vec X^{(1)}, \ldots, \vec X^{(k)} \vert \vec A^{(1)}, \ldots, \vec A^{(k)}]$ satisfying, for all $j_1, \ldots, j_k \geq 0$,  
\begin{equation*}
\Pr[\vec X^{(1)}_{1:j_1} \ldots \vec X^{(k)}_{1:j_k} \vert \vec A^{(1)} \ldots \vec A^{(k)}]
= 
\Pr[\vec X^{(1)}_{1:j_1} \ldots \vec X^{(k)}_{1:j_k} \vert \vec A^{(1)}_{1:j_1} \ldots \vec A^{(k)}_{1:j_k}]
 \enspace .
\end{equation*}
\end{definition}

This definition has two intuitive implications.  First is causality: the distribution of the outputs should only depend on the inputs already given, e.g., $\Pr[\vec X^{(1)}_{1:j_1} \vert \vec A^{(1)}] = \Pr[\vec X^{(1)}_{1:j_1} \vert \vec A^{(1)}_{1:j_1}]$.  Second is non-signaling: it does not matter in what order inputs are given.  For example, player~$k$'s inputs cannot change the marginal distribution of the other players' outputs, $\Pr[\vec X^{(1)}_{1:j_1} \ldots \vec X^{(k-1)}_{1:j_{k-1}} \vert \vec A^{(1)} \ldots \vec A^{(k)}] = \Pr[\vec X^{(1)}_{1:j_1} \ldots \vec X^{(k-1)}_{1:j_{k-1}} \vert \vec A^{(1)}_{1:j_1} \ldots \vec A^{(k-1)}_{1:j_{k-1}}]$; this follows by setting $j_k = 0$ in the definition, and intuitively corresponds to player~$k$ going last.  See \figref{f:multiroundns}.  

In general, multiple multi-round non-signaling resources between the same set of players can be grouped together into one.  This is useful because, unlike in the proofs of~\cite{BarrettPironio05nonsignaling, AnshuMhalla12NSgraph, HoyerMhallaPerdrix16graphgames} we cannot here eliminate one non-signaling resource at a time, because the players' strategy might use an unbounded number of resources.  Instead, we will use the non-signaling property to eliminate, all at once, all uses of non-signaling resources involving a given subset of $k$ players.  

\begin{figure}
\centering
\begin{tabular}{c@{$\qquad\qquad\qquad$}c}
\subfigure[]{\raisebox{.2cm}{\includegraphics[scale=.393]{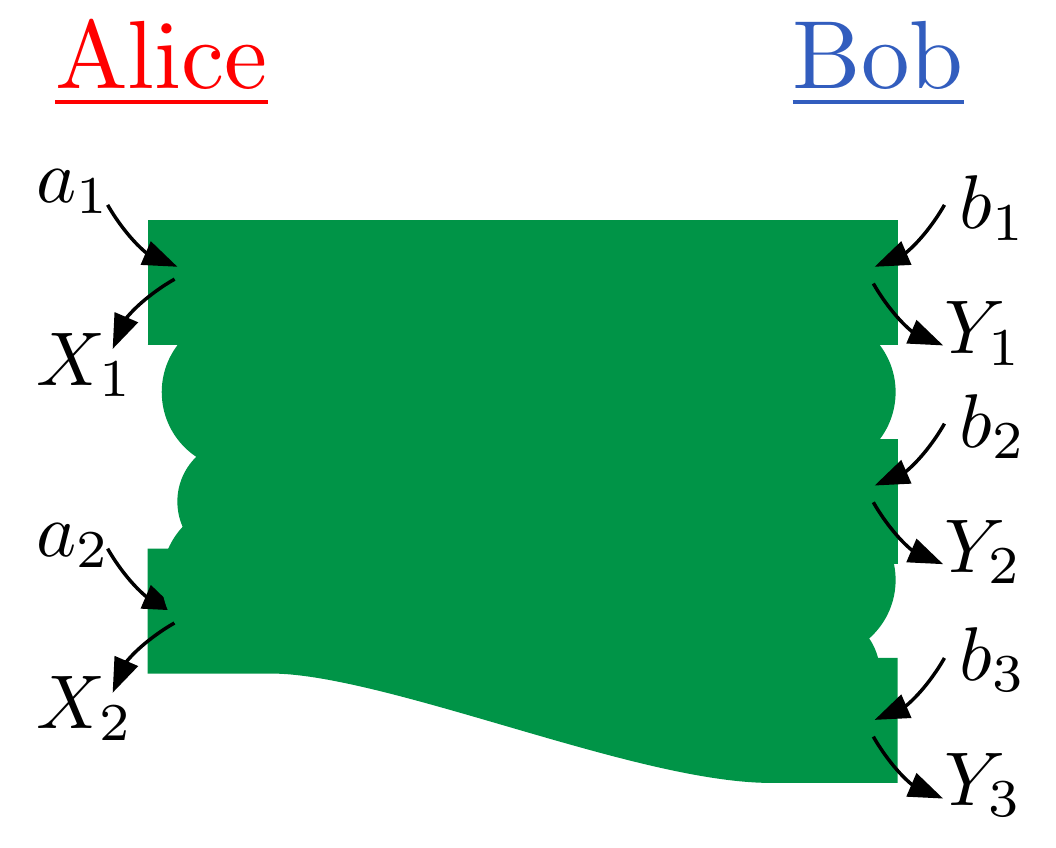}}} & 
\subfigure[]{\raisebox{.66cm}{\includegraphics[scale=.393]{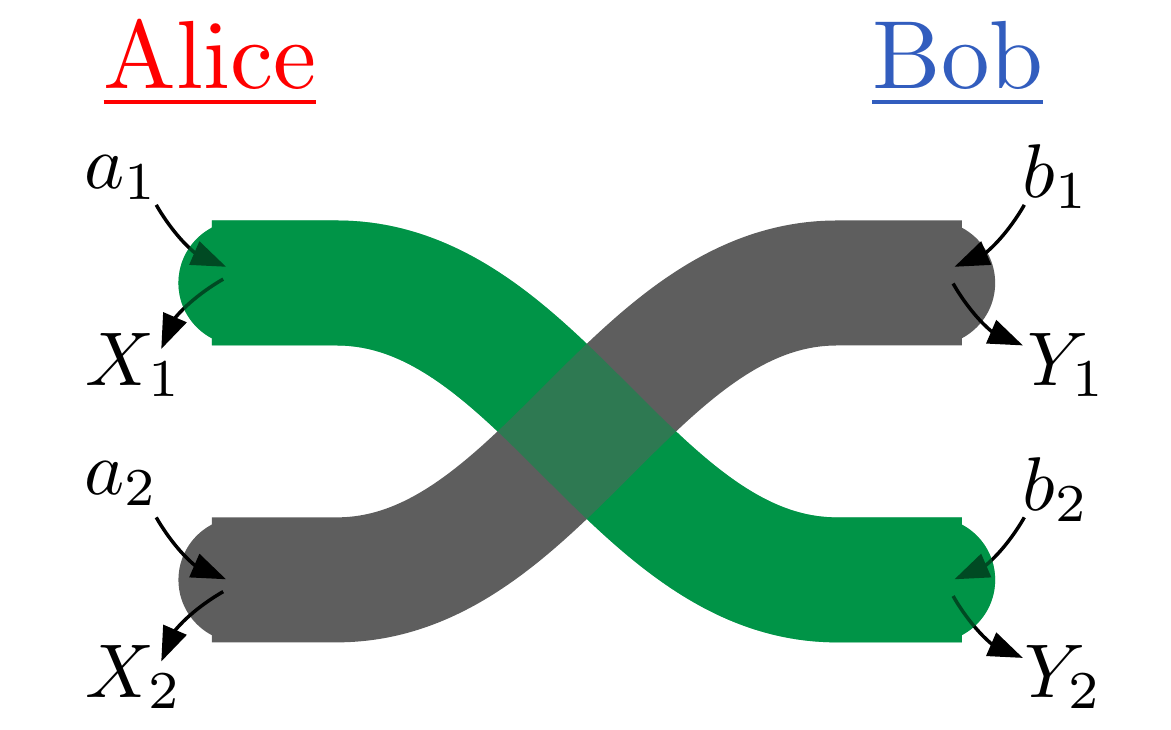}}} 
\end{tabular}
\caption{
(a) A multi-round non-signaling resource allows the players to make sequential, asynchronous queries.  Although the responses can be correlated, the correlations do not allow for signaling; the marginal distribution of \Alice's outputs is independent of \Bob's inputs, and vice versa.  
(b) For example, \Alice\ and \Bob\ might share two nonlocal boxes, but query them in opposite orders, so the outputs satisfy $X_1 \oplus Y_2 = a_1 b_2$ and $X_2 \oplus Y_1 = a_2 b_1$.  
} \label{f:multiroundns}
\end{figure}

\medskip

In general, a nonlocal game involves some set of players, who can agree in advance on a strategy but cannot later communicate with each other.  The verifier chooses from some distribution a question, consisting of a set of inputs to all or some of the players, and sends each player its input.  The players respond, and the verifier applies a predicate to decide whether to accept.  The game is ``unique" if for any question and any player~$v$ involved in the verifier's acceptance predicate, for any fixed responses from the other involved players there is exactly one response for~$v$ so that the verifier accepts.  

\begin{lemma} \label{t:fixrandomnessgeneral}
Consider a unique nonlocal game.  Let $\S$ be a strategy, for classical players using non-signaling resources, that wins with probability at least $1 - \epsilon$ on all questions.  Fix a player~$v$ and a non-signaling resource~$R$ involving~$v$, and assume that there exists a question~$\Q$ such that the verifier's acceptance predicate depends only on the responses of~$v$ and players $U$ not involved in~$R$.  Assume further that for any other question~$\Q'$ in which $v$'s input is the same as in $\Q$, for any player $u \in U$ either $u$'s input is the same as in $\Q$ or the verifier's acceptance predicate does not depend on $u$'s response.  

Then there exists a strategy~$\S'$ that wins with probability at least $1 - 3 \epsilon$ on all questions, and that is the same as~$\S$ except for $v$'s behavior on its input in question~$\Q$; on this input, $v$ ignores the resource~$R$.  
(In fact, $\Pr[\text{\emph{$\S'$ loses}} \vert \Q] \leq \Pr[\text{\emph{$\S$ loses}} \vert \Q]$ and $\Pr[\text{\emph{$\S'$ loses}} \vert \Q'] \leq \Pr[\text{\emph{$\S$ loses}} \vert \Q'] + 2 \Pr[\text{\emph{$\S$ loses}} \vert \Q]$ for any $\Q'$ in which $v$'s input is the same as in~$\Q$.)  
\end{lemma}

\begin{proof}
Let $U$ be the set of players, aside from~$v$, upon whose answers the verifier's acceptance predicate depends for question~$\Q$.  Let $W$ be the set of all players aside from $U$ and~$v$.  

Call a question $\Q'$ \emph{$v$-compatible} with $\Q$ if either $v$'s input differs from $\Q$ to $\Q'$ (so $\S$ and $\S'$ will be the same on $\Q'$), or for any player $u \in U$, either $u$'s input is the same as in $\Q$ or the verifier's acceptance predicate does not depend on $u$'s response.  By assumption all questions chosen with positive probability are $v$-compatible with~$\Q$.  

Parameterize the randomness for the non-signaling resources according to the players in~$U$ going first, then $v$, then $W$: 
\begin{enumerate}[(a)]
\item 
Let $r_U$ denote the randomness for $U$.  It fixes the answers of the players in~$U$ to~$\Q$.  Let $U(r_U)$ denote the unique answer for~$v$ for which the verifier accepts.  
\item 
Let $r_{Uv}$ denote the remaining randomness needed to determine the outputs at~$v$ of any resources that include $v$ and a player or players in~$U$ (these resources might also include players in~$W$).  Let $r_v$ denote the randomness needed to determine the outputs at~$v$ of any resources that include~$v$ but no players in~$U$.  Then $v$'s answer to its input~$\Q_v$ in $\Q$ is a deterministic function $v(r_U, r_{Uv}, r_v)$.  
\item 
Let $r_W$ denote all the remaining randomness for non-signaling resources involving~$W$; this includes resources that cross from $U$ and/or $v$ to $W$ as well as any other resources involving~$W$.  

For a $v$-compatible question $\Q'$, let $W_{\Q'}(r_U, r_{Uv}, r_v, r_W)$ denote the unique answer for $v$ for which the verifier accepts.  It is a deterministic function of the randomness that we have defined, where the inputs to $U$ and $v$ are given by $\Q$ and the inputs to $W$ given by $\Q'$.  
\end{enumerate}

Let $\chi_P$ denote the indicator function for a predicate~$P$; $\chi_P = 1$ if $P$ is true and $\chi_P = 0$ otherwise.  

Consider the question $\Q$.  Strategy $\S$ wins with probability at least $1-\epsilon$ on this question, i.e., with the above parameterization, 
\begin{equation} \label{e:uwinprobgeneral}
1 - \epsilon 
\leq 
\sum_{r_U, r_{Uv}, r_v} p(r_U) p(r_{Uv}) p(r_v) \cdot \chi_{U(r_U) = v(r_U, r_{Uv}, r_v)}
\end{equation}
In particular, there exists a fixed value $r_v^*$ such that 
\begin{equation} \label{e:uwinprobfixedrgeneral}
1 - \epsilon 
\leq 
\sum_{r_U, r_{Uv}} p(r_U) p(r_{Uv}) \cdot \chi_{U(r_U) = v(r_U, r_{Uv}, r_v^*)}
\end{equation}
Define strategy $\S'$ using $r_v^*$: on input $Z$, $v$ does not use any resource that does not also involve~$U$; for any such resource, $v$ instead simulates its input/output behavior using the fixed $r_v^*$.  

To prove the lemma, we need to lower bound the success probability of~$\S'$ on question $\Q$, and on other $v$-compatible questions.  

\medskip

1. For question $\Q$, the right-hand side of Eq.~\eqnref{e:uwinprobfixedrgeneral} is exactly the probability that $\S'$ wins; indeed it is at least $1 - \epsilon$.  

\medskip

2. Consider a $v$-compatible question $\Q'$.  For strategy $\S$, let $W_{\Q'}(r_U, r_{Uv}, r_v, r_W)$ denote the unique answer for $v$ for which the verifier accepts, as a deterministic function of the randomness that we have defined, where the inputs to $U$ are given by $\Q$.  Similarly define $W_{\Q'}'(r_U, r_{Uv}, r_v, r_W)$ for strategy~$\S'$.  Since $v$ behaves differently in $\S'$ (in particular ignoring $r_v$), it is important to recognize that $W_{\Q'}'$ could be very different from $W_{\Q'}$.  $W_{\Q'}'$ depends on both $r_v$ and~$r_v^*$.  

The probabilities that the strategies win satisfy 
\begin{align}
1 - \epsilon \leq \Pr[\text{$\S$ wins} \vert \Q'] 
&= \sum_{\substack{r_U, r_{Uv} \\ r_v, r_W}} p(r_U) p(r_{Uv}) p(r_v) p(r_W) \cdot \chi_{v(r_U, r_{Uv}, r_v) = W_{\Q'}(r_U, r_{Uv}, r_v, r_W)} \nonumber \\
&= \Pr[v = W_{\Q'}] \label{e:uwinprobWgeneral} \\
\Pr[\text{$\S'$ wins} \vert \Q'] 
&= \Pr[v^* = W_{\Q'}'] \nonumber
 \enspace .
\end{align}
We use the shorthand $v = v(r_U, r_{Uv}, r_v)$, $v^* = v(r_U, r_{Uv}, r_v^*)$, $W_{\Q'} = W_{\Q'}(r_U, r_{Uv}, r_v, r_W)$, $W_{\Q'}' = W_{\Q'}'(r_U, r_{Uv}, r_v, r_W)$ and $U = U(r_U)$.  

From $\chi_{v^* \neq W_{\Q'}'} \leq \chi_{v^* \neq U} + \chi_{U \neq W_{\Q'}'}$, we bound 
\begin{align*}
\Pr[\text{$\S'$ loses} \vert \Q']
&= \Pr[v^* \neq W_{\Q'}'] \\
&\leq \Pr[v^* \neq U] + \Pr[U \neq W_{\Q'}']
 \enspace .
\end{align*}
By the no-signaling property, $v$'s different actions in the strategy $\S'$ cannot affect the joint distributions of the players in $U$ and~$W$ together.  Therefore $\Pr[U \neq W_{\Q'}'] = \Pr[U \neq W_{\Q'}]$.  (This is the key observation in the proof.)  We conclude 
\begin{align*}
\Pr[\text{$\S'$ loses} \vert \Q']
&\leq \Pr[v^* \neq U] + \Pr[U \neq W_{\Q'}] \\
&\leq \Pr[v^* \neq U] + \Pr[U \neq v] + \Pr[v \neq W_{\Q'}] \\
&= \Pr[\text{$\S'$ loses} \vert \Q] + \Pr[\text{$\S$ loses} \vert \Q] + \Pr[\text{$\S$ loses} \vert \Q'] \\
&\leq 3 \epsilon
 \enspace ,
\end{align*}
where the last steps use Eqs.~\eqnref{e:uwinprobfixedrgeneral}, \eqnref{e:uwinprobgeneral} and~\eqnref{e:uwinprobWgeneral}, respectively, to bound the three terms.  
\end{proof}

\section{Robust separation argument} \label{s:robustseparation}

We will now apply \lemref{t:fixrandomnessgeneral} to prove an upper bound on the probability that classical players sharing non-signaling resources can win an extended game.  We give the proofs in parallel for the $\mathrm{CHSH} + k$ game (\secref{s:extendedchshgame}) and the $\mathrm{CHSH}_n + k$ game (\appref{s:extendedchshngame}); the latter gives a larger separation for $k > 1$.  

\begin{theorem} \label{t:robustseparationchsh}
For any $k$, classical players sharing arbitrary $(k+1)$-party non-signaling resources can win the $\mathrm{CHSH} + k$ game with probability at most $1 - \frac{1}{2 (3^k + 1)} < \valuequantum(\mathrm{CHSH} + k)$.  
\end{theorem}

\begin{theorem} \label{t:robustseparationchshn}
For any $k$, classical players sharing arbitrary $(k+1)$-party non-signaling resources can win the $\mathrm{CHSH}_n + k$ game with probability at most $1 - \frac{1}{2 (2 n + 3^k - 2)} < \valuequantum(\mathrm{CHSH}_n + k)$.  
\end{theorem}

\noindent
In particular, using \propref{t:distributedgamequantumvalue}, $\valuequantum(\mathrm{CHSH} + k) - (1 - \frac{1}{2(3^k+1)}) \geq \frac{\sqrt 2 - 1}{2 (3^k + 1)}$.  
Explicit lower bounds on the quantum versus non-signaling gap are listed in \figref{f:klocalgaps}.  

\begin{figure}
\centering
\begin{tabular}{c @{\;}|@{\;\;} c @{\quad} c @{\;} l}
\hline \hline
$k$ & CHSH gap & \multicolumn{2}{c}{Best CHSH$_n$ gap$\qquad$} \\
\hline
1 & $5.178 \cdot 10^{-2}$ & $4.272 \cdot 10^{-2}$ & (with $n = 3$) \\
2 & $2.071 \cdot 10^{-2}$ & $2.318 \cdot 10^{-2}$ & ($n = 4$) \\
3 & $7.397 \cdot 10^{-3}$ & $1.079 \cdot 10^{-2}$ & ($n = 5$) \\
4 & $2.526 \cdot 10^{-3}$ & $4.454 \cdot 10^{-3}$ & ($n = 8$) \\
5 & $8.488 \cdot 10^{-4}$ & $1.695 \cdot 10^{-3}$ & ($n = 13$) \\
6 & $2.837 \cdot 10^{-4}$ & $6.122 \cdot 10^{-4}$ & ($n = 22$) \\
\hline \hline
\end{tabular}
$\quad$
\begin{tabular}{c @{\;}|@{\;\;} c @{\quad} c @{\;} l}
\hline \hline
$k$ & CHSH gap & \multicolumn{2}{c}{Best CHSH$_n$ gap$\qquad$} \\
\hline
7 & $9.466 \cdot 10^{-5}$ & $2.140 \cdot 10^{-4}$ & (with $n = 38$) \\
8 & $3.156 \cdot 10^{-5}$ & $7.333 \cdot 10^{-5}$ & ($n = 65$) \\
9 & $1.052 \cdot 10^{-5}$ & $2.484 \cdot 10^{-5}$ & ($n = 111$) \\
10 & $3.507 \cdot 10^{-6}$ & $8.359 \cdot 10^{-6}$ & ($n = 192$) \\
11 & $1.169 \cdot 10^{-6}$ & $2.802 \cdot 10^{-6}$ & ($n = 332$) \\
12 & $3.897 \cdot 10^{-7}$ & $9.368 \cdot 10^{-7}$ & ($n = 574$) \\
\hline \hline
\end{tabular}
\caption{Gap lower bound between the quantum and $(k+1)$-local non-signaling strategies for the $\mathrm{CHSH} + k$ and $\mathrm{CHSH}_n + k$ games.} \label{f:klocalgaps}
\end{figure}

\newcommand{\cR}{{\mathcal R}}

\begin{proof}[Proof of Theorems~\ref{t:robustseparationchsh} and~\ref{t:robustseparationchshn}]
The proofs for $\mathrm{CHSH} + k$ and $\mathrm{CHSH}_n + k$ are the same except for the algebra at the end.  

Let $\S$ be any classical strategy for the game, $\mathrm{CHSH} + k$ or $\mathrm{CHSH}_n + k$, using $(k+1)$-party non-signaling resources.  Assign each non-signaling resource involving $< k+1$ players to an arbitrary superset of $k+1$ players.  (For example, a two-local resource shared between \Alice\ and \Bob\ might be assigned to $\{\mathrm{\Alice}, \mathrm{\Bob}, \text{\Charliej{1}}, \ldots, \text{\Charliej{k-1}}\}$.)  Then for each subset of $k+1$ players, group together all the associated correlations into a single multi-round non-signaling resource.  Therefore we may assume that $\S$ uses exactly $k+2$ multi-round non-signaling resources.  Denote by $\cR_\A$ the resource involving all players \emph{except} \Alice, by $\cR_\B$ the resource involving all players except \Bob, and by $\cR_\Cj{j}$ the resource involving all players except \Charliej{j}.  

Let $\Q_1, \ldots, \Q_k$ be the $k$ consistency questions, and let $\Q_{a,b}$ be the game question in which \Alice\ and \Bob's respective inputs are $a$ and~$b$.  
For any question $\Q$, let $\epsilon_\Q = \Pr[\text{$\S$ loses} \vert \Q]$.  

Begin by considering $\Q_1$, the consistency question between \Alice\ and \Charliej{1}.  Their inputs are $A = C_1 = 0$.   We aim to apply \lemref{t:fixrandomnessgeneral} for \Alice\ and resource~$\cR_\Cj{1}$.  The two main assumptions of the lemma hold.  Indeed, of these two players, only \Alice\ has access to $\cR_\Cj{1}$.  Furthermore, no other question $\Q'$ with $A = 0$, either a consistency question or a game question, depends on \Charliej{1}'s output $Z_1$.  (Although $\mathrm{CHSH} + k$ and $\mathrm{CHSH}_n + k$ are not unique games, they are unique for all questions with $A = 0$.)  Therefore, by \lemref{t:fixrandomnessgeneral}, there exists a strategy $\S_1$ that is the same as $\S$ except that \Alice\ on input $A = 0$ ignores $\cR_\Cj{1}$.  The loss probabilities $\epsilon^{(1)}_\Q = \Pr[\text{$\S_1$ loses} \vert \Q]$ satisfy 
\begin{align*}
\epsilon^{(1)}_{\Q_1} &\leq \epsilon_{\Q_1} &
\epsilon^{(1)}_{\Q_{0,b}} &\leq \epsilon_{\Q_{0,b}} + 2 \epsilon_{\Q_1} \\
\epsilon^{(1)}_{\Q_j} &\leq \epsilon_{\Q_j} + 2 \epsilon_{\Q_1} \quad\text{for $j \neq 1$} &
\epsilon^{(1)}_{\Q_{a,b}} &= \epsilon_{\Q_{a,b}} \qquad\text{for $a \neq 0$}
 \enspace .
\end{align*}

Repeat the above argument for question $\Q_2$.  Applying \lemref{t:fixrandomnessgeneral} for \Alice\ and resource~$\cR_\Cj{2}$, we obtain a strategy $\S_2$ in which on input $A = 0$ \Alice\ ignores resources $\cR_\Cj{1}$ and $\cR_\Cj{2}$.  The loss probabilities satisfy $\epsilon^{(2)}_{\Q_{a,b}} = \epsilon_{\Q_{a,b}}$ for $a \neq 0$, and 
\begin{alignat*}{3}
\epsilon^{(2)}_{\Q_j} &\leq \epsilon^{(1)}_{\Q_j} + 2 \epsilon^{(1)}_{\Q_2} 
&&\leq \epsilon_{\Q_j} + 6 \epsilon_{\Q_1} + 2 \epsilon_{\Q_2}
&&\quad\text{for $j > 2$} \\
\epsilon^{(2)}_{\Q_{0,b}} &\leq \epsilon^{(1)}_{\Q_{0,b}} + 2 \epsilon^{(1)}_{\Q_2} 
&&\leq \epsilon_{\Q_{0,b}} + 6 \epsilon_{\Q_1} + 2 \epsilon_{\Q_2}
\end{alignat*}

Continue inductively.  Ultimately, we construct a strategy $\S_k$ in which on input $A = 0$ \Alice\ ignores all the resources $\cR_\Cj{1}, \ldots, \cR_\Cj{k}$, and for which $\epsilon^{(k)}_{\Q_{a,b}} = \epsilon_{\Q_{a,b}}$ for $a \neq 0$, and \begin{equation*}
\epsilon^{(k)}_{\Q_{0,b}}
\leq \epsilon_{\Q_{0,b}} + 2 \, \big( 3^{k-1} \epsilon_{\Q_1} + 3^{k-2} \epsilon_{\Q_2} + \cdots + \epsilon_{\Q_k} \big)
 \enspace .
\end{equation*}

Similar inequalities hold for the strategies one obtains by eliminating the consistency questions in any other order.  Averaging over the $k$ cyclically permuted orderings yields a strategy $\bar \S_k$ for which the loss probabilities $\bar \epsilon^{(k)}_\Q = \Pr[\text{$\bar \S_k$ loses} \vert \Q]$ satisfy $\bar \epsilon^{(k)}_{\Q_{a,b}} = \epsilon_{\Q_{a,b}}$ for $a \neq 0$, and 
\begin{align*}
\epsilon^{(k)}_{\Q_{0,b}}
&\leq \epsilon_{\Q_{0,b}} + \frac{3^k-1}{k} ( \epsilon_{\Q_1} + \cdots + \epsilon_{\Q_k} ) \\
&= \epsilon_{\Q_{0,b}} + (3^k - 1) \Pr[\text{$\S$ loses} \vert \text{consistency}]
 \enspace .
\end{align*}

From here on we consider only game questions $\Q_{a,b}$.  On input $A = 0$, \Alice\ in $\bar \S_k$ only uses the resource $\cR_\B$, which does not involve \Bob.  Since in this case the verifier's acceptance predicate $X = Y$ depends only on \Alice\ and \Bob, $\cR_\B$ intuitively should not be helpful.  While \lemref{t:fixrandomnessgeneral} does not apply, the argument is straightforward.  

Consider the questions $\Q_{0,b}$ with $b \in \{0,1\}$.  \Alice\ queries only $\cR_\B$ and \Bob\ queries only those resources aside from $\cR_\B$.  By the non-signaling property, we may assume that the other players $\text{\Charliej{1}}, \ldots, \text{\Charliej{k}}$ go last, and thus can parameterize \Alice\ and \Bob's local randomness by $r_\A$ and~$r_\B$, respectively, such that their outputs are deterministic functions $f(r_\A, z)$ and $g(b, r_\B, z)$ of their inputs and local randomness, and any shared randomness~$z$.  Thus the probability of $\bar \S_k$ winning on a game question with $A = 0$ satisfies 
\begin{align*}
\sum_{b=0}^1 \frac12 \Pr[\text{$\bar S_k$ wins} \vert \Q_{0,b}]
&= \sum_{z, r_\A, b, r_\B} p(z) p(r_\A) p(b) p(r_\B) \cdot \chi_{f(r_\A, z) = g(b, r_\B, z)} \\
&= \sum_z p(z) \sum_{c=0}^1 \Pr\!\big[\{ r_\A : f(r_\A, z) = c \}\big] \Pr\!\big[\{ (b, r_\B) : g(b, r_\B, z) = c \}\big]
 \enspace .
\end{align*}
For each $z$, let $c_z \in \{0,1\}$ to be the value that maximizes $\Pr\!\big[\{ (b, r_\B) : g(b, r_\B, z) = c \}\big]$.  Define a strategy $\S'$ for the players that is the same as $\bar \S_k$, except that \Alice\ always outputs $c_z$ when her input is $A = 0$.  

\begin{claim} \label{t:markov}
$\Pr[\text{\emph{$\S'$ loses}} \vert \text{\emph{game question}}] 
\leq \big( 1 + (3^k - 1) \Pr[A = 0 \vert \text{\emph{game}}] \big) \Pr[\text{\emph{$\S$ loses}}]$.  
\end{claim}

\begin{proof}
On game questions with $A = 0$, $\S'$ wins with at least the probability that $\bar \S_k$ wins; and on game questions with $A \neq 0$, $\S'$, $\bar \S_k$ and $\S$ are the same.  Therefore on game questions we have, for $p = \Pr[A = 0 \vert \text{game question}]$, 
\begin{align*}
\Pr[\text{$\S'$ loses} \vert \text{game question}] 
&= p \Pr[\text{$\S'$ loses} \vert \text{game, $A = 0$}] + (1 - p) \Pr[\text{$\S'$ loses} \vert A \neq 0] \\
&\leq p \Pr[\text{$\bar \S_k$ loses} \vert \text{game, $A = 0$}] + (1 - p) \Pr[\text{$\S$ loses} \vert A \neq 0]
 \enspace .
\end{align*}
We want to relate this bound to $\Pr[\text{$\S$ loses}]$.  Substitute 
\begin{align*}
\Pr[\text{$\bar \S_k$ loses} \vert \text{game, $A = 0$}]
&= \tfrac12 \big( \Pr[\text{$\bar \S_k$ loses} \vert \Q_{0,0}] + \Pr[\text{$\bar S_k$ loses} \vert \Q_{0,1}] \big) \\
&\leq \Pr[\text{$\S$ loses} \vert \text{game, $A = 0$}] + (3^k - 1) \Pr[\text{$\S$ loses} \vert \text{consistency}] 
 \enspace .
\end{align*}
Then use $\Pr[\text{$\S$ loses}] = (1 - q) \Pr[\text{$\S$ loses} \vert \text{game}] + q \Pr[\text{$\S$ loses} \vert \text{consistency}]$, where $q = 1 / \big( 1 + \tfrac{1}{(3^k - 1) p} \big)$ is the probability of choosing a consistency question: 
\begin{align*}
\Pr[\text{$\S'$ loses} \vert \text{game question}] 
&\leq \Pr[\text{$\S$ loses} \vert \text{game}] + p (3^k - 1) \Pr[\text{$\S$ loses} \vert \text{consistency}] \\
&= \frac{1}{1-q} (1-q) \Pr[\text{$\S$ loses} \vert \text{game}] + \frac{(3^k - 1) p}{q} q \Pr[\text{$\S$ loses} \vert \text{consistency}] \\
&= \big( 1 + (3^k - 1) p \big) \Pr[\text{$\S$ loses}]
 \enspace ,
\end{align*}
since our choice for $q$ balances the coefficients, $\tfrac{1}{1 - q} = \tfrac{(3^k-1)p}{q} = 1 + (3^k-1)p$.  
\end{proof}

\paragraph{Analysis for $\mathrm{CHSH} + k$ game.}

Substitute into \claimref{t:markov} $\Pr[\text{$\S'$ loses} \vert \text{game}] \geq 1/4$ (from \propref{t:chshdeterministica0}) and $p = 1/2$ to find $\Pr[\text{$\S$ loses}] \geq \tfrac{1}{2 (3^k + 1)}$.

\paragraph{Analysis for $\mathrm{CHSH}_n + k$ game.}

We will argue that $\S'$ wins one of the two embedded CHSH$_n$ games with high probability, and apply \propref{t:chshndeterministica0}.  Having restricted to game questions, we may assume $C_1 = \cdots = C_k = 1$.  The inputs of the \Charlie\ players do not vary.  Therefore, \Alice\ and \Bob\ can simulate their interactions with the \Charlie\ players using shared randomness.  We can fix a value for this and all shared randomness.  Then either $Z = 0$ or $Z = 1$.  In the former case, the subset of questions $A \in \{0, \ldots, n-1\}$ forms a CHSH$_n$ game, while in the latter case the subset of questions $A \in \{1-n, \ldots, 0\}$ does.  

Consider the $Z = 0$ case; the case $Z = 1$ is symmetrical.  We have 
\begin{align*}
\Pr[\text{$\S'$ loses} \vert \text{game}] 
&= \Pr[\text{$\S'$ loses} \vert \text{$A \geq 0$, game}] \Pr[A \geq 0 \vert \text{game}] + \Pr[\text{$\S'$ loses} \vert A < 0] \Pr[A < 0 \vert \text{game}] \\
&= \Pr[\text{$\S'$ loses} \vert \text{$A \geq 0$, game}] \Pr[A \geq 0 \vert \text{game}]
 \enspace .
\end{align*}
(The last step is an equality, not just $\geq$, because $Z = 0$ wins the game when $A < 0$.)  

Conditioned on $A \geq 0$, we have $\Pr[A = a \vert \text{$A \geq 0$, game}] = \tfrac{1}{n}$ for $a \in \{0, \ldots, n-1\}$, the input probabilities for the CHSH$_n$ game.  On input $A = 0$, \Alice's output is deterministic.  By \propref{t:chshndeterministica0}, $\Pr[\text{$\S'$ loses} \vert \text{$A \geq 0$, game}] \geq \tfrac{1}{2 n}$.  It follows that $\Pr[\text{$\S'$ loses} \vert \text{game}] \geq 1 / (4 n - 2)$.  

Thus 
\begin{align*}
\Pr[\text{$\S$ loses}]
\geq \frac{1}{1 + \tfrac{3^k - 1}{2n-1}} \frac{1}{4n - 2}
= \frac{1}{2 (2 n + 3^k - 2)}
 \enspace .
\end{align*}
This is $\Omega(1/n)$ if we consider $k$ a constant.  For sufficiently large $n$, this is indeed worse than the optimal quantum strategy, which loses with probability less than $\tfrac{1}{2 n^2}$ (\propref{t:extendedchshnquantumvalue}).  
\end{proof}

\section{Open problems} \label{s:openproblems}

A natural open problem is to improve the gap between quantum and two-local non-signaling theories, in order to ease experimental tests.  We do not know whether our analysis for the Extended CHSH game is tight.  Can parallel repetition help?  Also, other games might give larger gaps, with similar experimental complexities.  For example, quantum players can win the five-player cycle graph game with probability one, whereas in unpublished work we have calculated using \lemref{t:fixrandomnessgeneral} that classical players sharing arbitrary two-local non-signaling correlations can win with probability at most $1 - \tfrac{5}{414} \approx 98.8\%$.  

A related research direction is to devise a two-party game for which there is a robust separation in the number/dimension of non-signaling resources required to win, versus the number of EPR states.  Broadbent and M\'ethot~\cite{BroadbentMethot06nonlocal} give an exponential separation, but it is not robust.

\subsection*{Acknowledgements}

We would particularly like to thank Zeph Landau and Umesh Vazirani for helpful conversations.  
Research supported by NSF grant CCF-1254119, ARO grant W911NF-12-1-0541, and the AFOSR.

\bibliographystyle{alpha-eprint}
\bibliography{q}

\appendix

\section{Nonlocal games that do not separate quantum from two-local non-signaling correlations} \label{s:distributedgamecounterexample}

Quantum players can win the Extended CHSH$_n$ game with a higher probability than classical players sharing two-local non-signaling correlations (\thmref{t:distributedgame}).  One way to understand this game and its analysis is to study games that lack this property, i.e., for which a two-local non-signaling strategy can match any quantum strategy.  

\smallskip

One of the simplest non-trivial three-party games is the GHZ game~\cite{GreenbergerHorneShimonyZeilinger90ghzgame}.  In this game, triples of valid inputs to the players are $(0,0,1), (0,1,0), (1,0,0)$ and $(1,1,1)$.  The verifier checks that the exor of the players' responses is $0$ in the first three cases, or~$1$ for the inputs $(1,1,1)$.  The classical value is $\valueclassical = 3/4$, and the quantum value is $\valuequantum = 1$.  However, if two of the players share a nonlocal box, then they can also win with probability~$1$~\cite{BroadbentMethot06nonlocal}.  (They give their inputs to the box and return its outputs, while the third player outputs~$0$.)  

In fact, any game in which the verifier's acceptance predicate depends only on the exor of the players' responses---any ``exor game"---can be won with certainty by classical players sharing nonlocal boxes~\cite{Dam05nonlocal}.  (Indeed, the idea is simple: if $x = \bigoplus_j x_j$ and $y = \bigoplus_j y_j$, where player~$j$ has bits~$x_j$ and~$y_j$, then $x y = \bigoplus_{i,j} x_i y_j$, and each term can be computed in a distributed fashion using a nonlocal box.  Therefore the players can compute any function on distributed bits.)  Exor games cannot separate quantum from two-local non-signaling strategies.  

\smallskip

\begin{figure}
\centering
\begin{tabular}{c@{$\qquad\qquad\qquad$}c}
\subfigure[\label{f:distributedgame}]{\raisebox{.2cm}{\includegraphics[scale=.125]{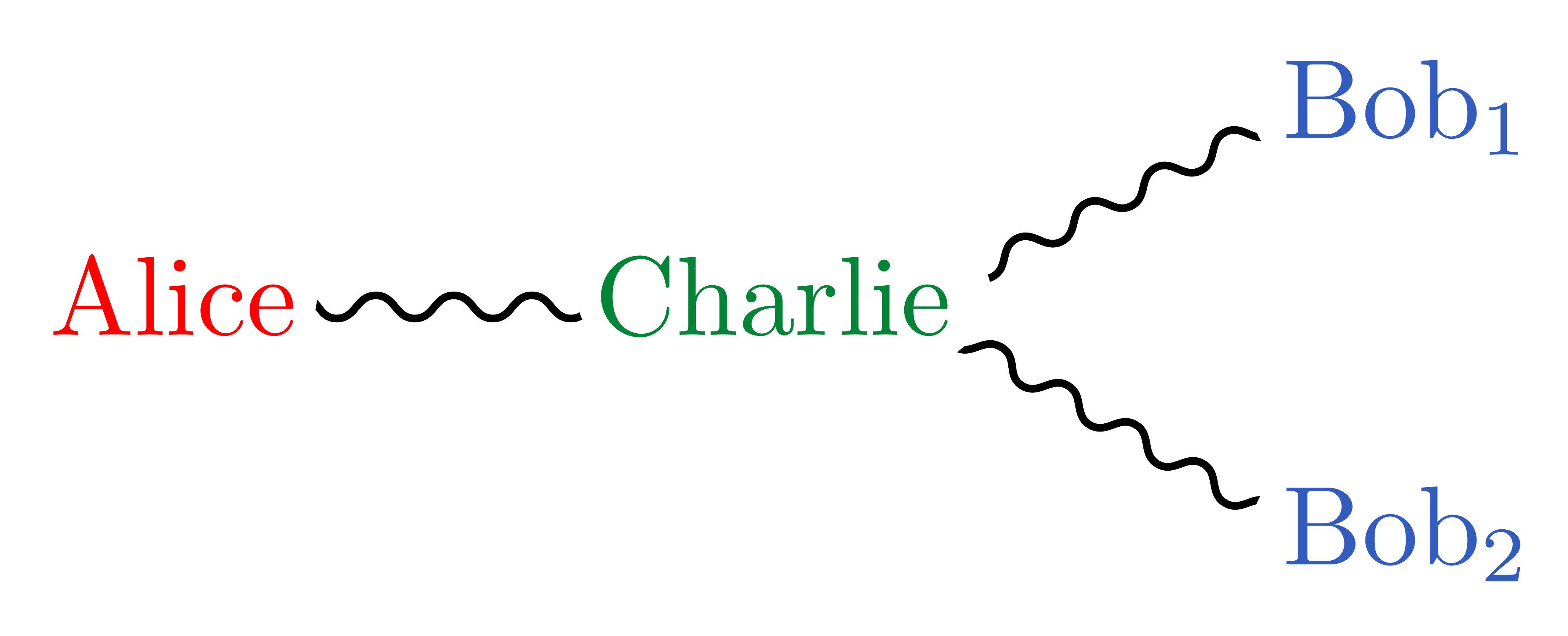}}} &
\subfigure[\label{f:distributedgamecondensed}]{\raisebox{1cm}{\includegraphics[scale=.125]{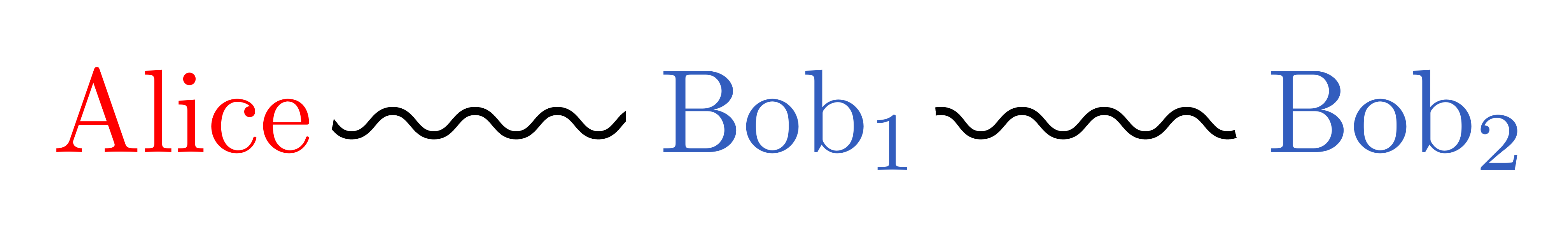}}}
\end{tabular}
\caption{
(a) In a distributed version of the CHSH game, \Alice\ plays a CHSH game with either \Bobone\ or \Bobtwo.  She does not know which, but \Charlie\ is told, and with Bell measurements he creates the necessary entanglement between \Alice\ and the selected \Bobj{J}.  
(b) A condensed version of the game.  Again, \Alice\ plays either with \Bobone\ or \Bobtwo, and in the latter case \Bobone\ teleports the needed quantum correlation.  
} \label{f:distributedgames}
\end{figure}

A game that intuitively might separate quantum from two-local non-signaling correlations is shown in \figref{f:distributedgame}.  \Charlie\ shares an EPR state with each of the other three players, \Alice, \Bobone\ and \Bobtwo.  The verifier tells \Charlie\ a random index $J \in \{1, 2\}$, so that he can use a Bell measurement to create an EPR state between \Alice\ and \Bobj{J}.  The verifier then referees a CHSH game between \Alice\ and \Bobj{J}, adjusting \Alice's answer according to the Pauli correction reported by \Charlie, as in Eq.~\eqnref{e:teleportedCHSH}.  The optimal quantum strategy wins with probability $\valuequantum(\mathrm{CHSH}) = \cos^2 \tfrac{\pi}{8}$.  

This distributed CHSH game is similar to the Teleported CHSH game from \thmref{t:teleport}, except with two \Bob\ players.  Intuitively, perhaps, two-local non-signaling correlations should not help the players, since \Alice\ does not know with which \Bob\ she is playing.  However, this intuition is incorrect: 

\begin{claim}
Players with appropriate two-local non-signaling correlations can win the distributed CHSH game with certainty.  
\end{claim}

\begin{proof}
\Alice\ uses two nonlocal boxes to play CHSH games with \Bobone\ and \Bobtwo, giving the same input~$A$ to both.  Call the respective outputs $X_1$ and~$X_2$.  She feeds these outputs into the following ``selection box" shared with \Charlie: 
\begin{equation*}
\includegraphics[scale=.45]{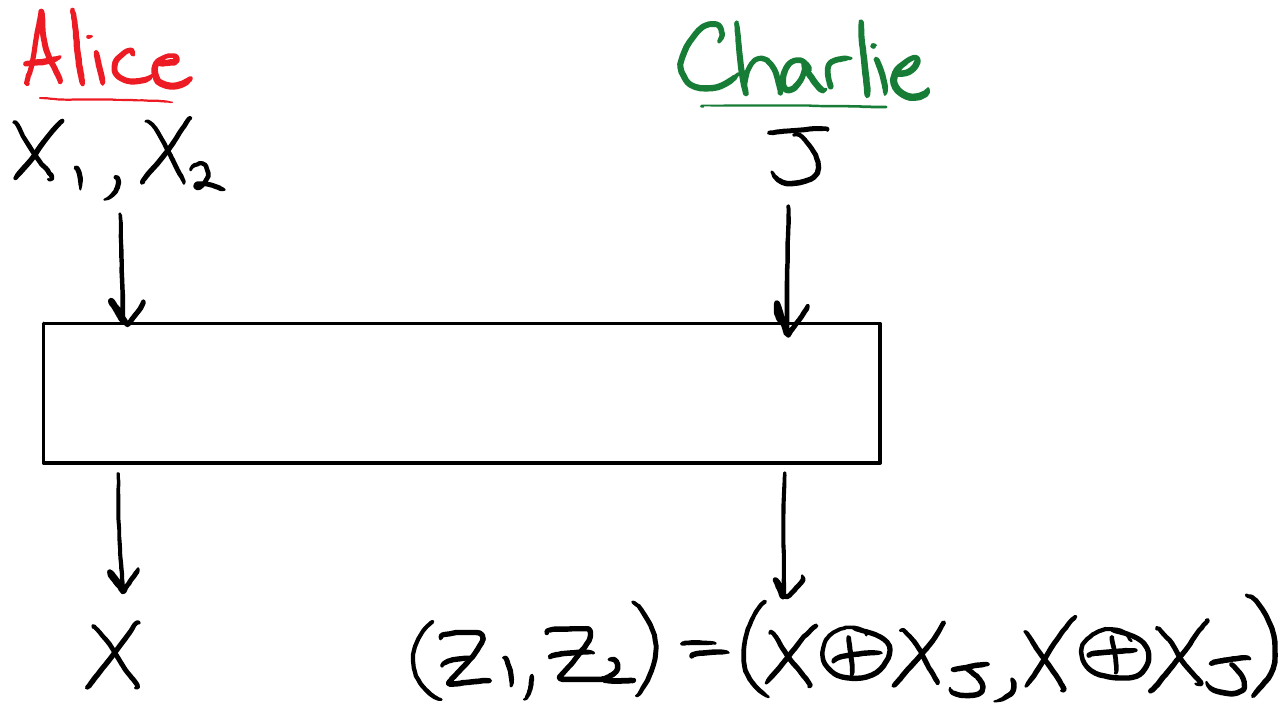}
\end{equation*}
\Alice's output is a uniformly random bit~$X$, but correlated such that its exor with either of \Charlie's $Z$ outputs equals $X_J$.  
\end{proof}

This distributed CHSH game superficially seems similar to the Extended CHSH game of \thmref{t:distributedgame}, except that the desired EPR state between \Alice\ and the selected player~$J$ is created by quantum teleportation instead of by a $\sigma^x$ measurement on $\ket{000} + \ket{111}$.  In using teleportation to create an EPR state, \Alice's Pauli correction can be any of $I, \sigma^x, \sigma^y, \sigma^z$.  On the other hand, a $\sigma^x$ measurement on $\ket{000} + \ket{111}$ generates an EPR state up to a correction only of $I$ or~$\sigma^z$, so the correction can be ignored when \Alice\ measures in the~$\sigma^z$ basis.  This difference is crucial.  

The first step in the proof of \thmref{t:distributedgame} also works for the distributed CHSH game.  That is, considering $J = 2$, we can argue that \Alice's answer depends little on $\randlocal{\A}{\Bj1}$, the local randomness in her resource with \Bobone.  However, we cannot follow the argument further, because the verifier's acceptance predicate always depends on \Charlie's answer.  (Although \Alice's final {answer} depends little on $\randlocal{\A}{\Bj1}$, her full transcript can depend on this randomness, and she can pass the dependence over to \Charlie.)  At best, one can bound below~$1$ the maximum success probability of a \emph{nonadaptive} two-local non-signaling strategy, i.e., a strategy in which the inputs to all resources must be decided on before receiving any outputs.  Furthermore, there exists a nonadaptive three-local non-signaling strategy that wins with certainty: 

\begin{claim}
Players with appropriate three-party non-signaling correlations can access them nonadaptively to win the distributed CHSH game with certainty.  
\end{claim}

\begin{proof}
Let $R$ be the following \Alice-\Charlie-\Bob\ resource.  \Charlie\ can input either ``yes" or ``no."  
\begin{itemize}
\item 
If \Charlie\ inputs ``yes," then the resource plays a Teleported CHSH game between the three parties.  Thus the outputs $X, (Z_1, Z_2), Y$ satisfy Eq.~\eqnref{e:teleportedCHSH}, with a uniformly random marginal distribution for any pair of the outputs, $X$ and $(Z_1, Z_2)$, $X$ and $Y$, or $(Z_1, Z_2)$ and $Y$.  
\item 
If \Charlie\ inputs ``no," then the resource gives the same uniformly random bit to \Alice\ and \Charlie, $X = Z_1 = Z_2$, and an independent uniformly random bit~$Y$ to \Bob.  
\end{itemize}
This resource is non-signaling.  

Assume that the players share two copies of~$R$, between \Alice-\Charlie-\Bobone\ and \Alice-\Charlie-\Bobtwo.  Each \Bob\ gives his input to his resource and outputs the response.  On input $J$, \Charlie\ gives ``yes" to the \Alice-\Charlie-\Bobj{J} resource and ``no" to the other resource.  He adds the outputs mod~$2$.  \Alice\ gives her input to both of her resources, and outputs the summed responses mod~$2$.  

\begin{figure}
\centering
\includegraphics[scale=.5]{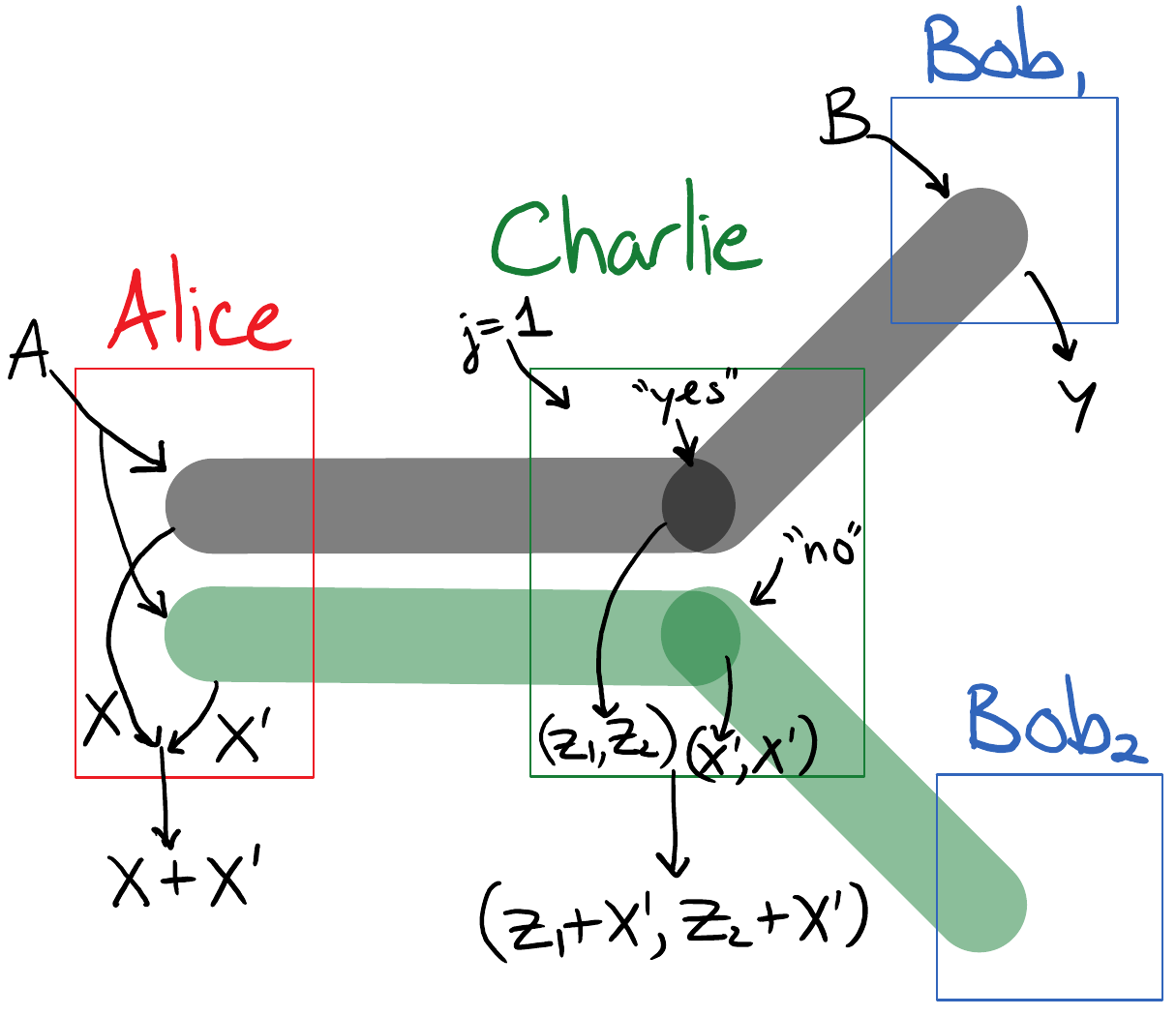}
\caption{Players sharing three-party non-signaling correlations can access them nonadaptively to win the distributed CHSH game.} \label{f:distributedchsh3localstrategy}
\end{figure}

For the $J = 1$ case, this strategy is illustrated in \figref{f:distributedchsh3localstrategy}.  When $A = 0$, the outputs satisfy 
\begin{align*}
\big(\text{\Alice's output $X + X'$}\big) + (\text{\Bobone's output $Y$}) + \big(\text{\Charlie's second output $Z_2 + X'$}\big) 
&= X + Y + Z_2 \\
&= A B
 \enspace ,
\end{align*}
whereas when $A = 1$ we have 
\begin{align*}
\big( X + X' \big) + Y + \big( Z_1 + X' \big) 
&= X + Y + Z_1 \\
&= A B
 \enspace ,
\end{align*}
since the $X'$ terms from the second shared resource cancel.  Therefore the players always win.  
\end{proof}

A similar strategy also works when there are more than two \Bob s, with \Alice-\Charlie-\Bobj{i} shared resources for all~$i$.  \Alice\ gives the same input~$A$ to all her resources and adds the responses; while \Charlie\ inputs ``yes" to the \Alice-\Charlie-\Bobj{j} resource, for the specified~$j$, ``no" to the others, and sums the responses.  Then \Alice\ and \Charlie's answers from the ``no" resources cancel out.

\section{Quantum versus non-signaling for Extended CHSH$_n$ games} \label{s:chshnanalysis}

In this section we consider generalizations of the CHSH game: CHSH$_n$ for $n = 2, 3, 4, \ldots$, where CHSH$_2$ is the standard CHSH game.  We define the $k$-Extended CHSH$_n$ game, for which \thmref{t:robustseparationchshn} establishes a gap between the winning probabilities achievable quantumly and those achievable by classical players sharing arbitrary $(k+1)$-local non-signaling resources.

\subsection{CHSH$_n$ game} \label{s:chshn}

The CHSH game is the first in a family of games known as chained Bell correlations~\cite{Pearle70chainedBellinequalities, BraunsteinCaves90chainedBellinequalities}.  

In CHSH$_n$, draw $A$ uniformly from $\{0, 1, \ldots, n-1\}$ and choose~$B$ either $A$ or $A + 1 \pmod n$ with equal probabilities.  \Alice\ takes $A$ and outputs a bit~$X$.  \Bob\ takes $B$ and outputs a bit~$Y$.  If $A = B = n - 1$, then accept if $X \neq Y$; and otherwise accept if $X = Y$.  For $n \geq 2$, the classical and quantum values are 
\begin{equation}\begin{split} \label{e:chainedbellvalues}
\valueclassical(\mathrm{CHSH}_n) &= 1 - \frac{1}{2 n} \\
\valuequantum(\mathrm{CHSH}_n) &= \cos^2 \frac{\pi}{4 n}
 \enspace .
\end{split}\end{equation}
An optimal quantum strategy uses one EPR state $\tfrac{1}{\sqrt 2}(\ket{00} + \ket{11})$.  On input~$a$, \Alice\ measures $\cos(a \tfrac{\pi}{n}) \sigma^z + \sin(a \tfrac{\pi}{n}) \sigma^x$ for $a = 0, \ldots, n-2$, or $\cos(-\tfrac{\pi}{n}) \sigma^z + \sin(-\tfrac{\pi}{n}) \sigma^x$ for $a = n-1$.  On input~$b$, \Bob\ measures $\cos((b-\tfrac12) \tfrac{\pi}{n}) \sigma^z + \sin((b-\tfrac12) \tfrac{\pi}{n}) \sigma^x$.  See \figref{f:chshn}.  
On all valid inputs, the success probability is $\cos^2 \tfrac{\pi}{4 n}$.  

\begin{figure}
\centering
\includegraphics[scale=1]{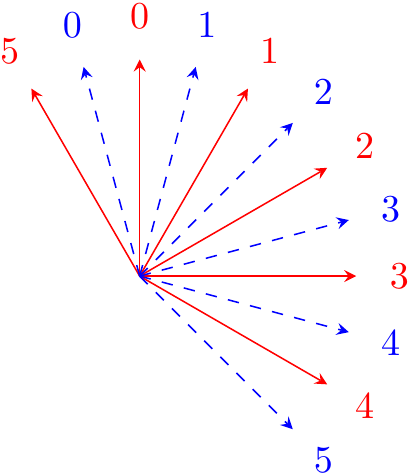}
\caption{Optimal measurement directions on the Bloch sphere for \Alice\ (solid, red) and \Bob\ (dashed, blue) for the CHSH$_n$ game with $n = 6$.  Consecutive vectors are separated by angle $\pi/(2n)$.} \label{f:chshn}
\end{figure}

The general CHSH$_n$ games have found frequent applications in quantum cryptography, e.g.,~\cite{BarrettHardyKent04diqkd, BarrettKentPironio06nonlocal}, and in quantum foundations~\cite{ColbeckRenner08hiddenvariables, ColbeckRenner11quantumcomplete, ColbeckRennet11amplifyrandomness}.  As for the CHSH game, \propref{t:chshdeterministica0}, an arbitrary non-signaling strategy in which $\Pr[X = 0 \vert A = 0] = 1$ cannot beat the classical value: 

\begin{proposition} \label{t:chshndeterministica0}
In the $\mathrm{CHSH}_n$ game with \Alice's response to question $A = 0$ fixed to~$X = 0$, the non-signaling value is $1 - \tfrac{1}{2n}$, the same as the classical value.  
\end{proposition}

\begin{proof}
For $a, b \in \{0, \ldots, n-1\}$, let 
\begin{align*}
\alpha_{ab} &= \Pr[X = 0, Y = 0 \vert A = a, B = b] &
\beta_{ab} &= \Pr[X = 0, Y = 1 \vert A = a, B = b] \\
\gamma_{ab} &= \Pr[X = 1, Y = 0 \vert A = a, B = b] &
\delta_{ab} &= \Pr[X = 1, Y = 1 \vert A = a, B = b]
 \enspace .
\end{align*}
These probabilities of course satisfy $\alpha_{ab} + \beta_{ab} + \gamma_{ab} + \delta_{ab} = 1$, as well as the non-signaling conditions: 
\begin{equation}\begin{aligned} \label{e:chshnonsignalingconditions}
\alpha_{jj} + \beta_{jj} &= \alpha_{j,j+1} + \beta_{j,j+1} = \Pr[X = 0 \vert A = j] \\
\alpha_{jj} + \gamma_{jj} &= \alpha_{j-1,j} + \gamma_{j-1,j} = \Pr[Y = 0 \vert B = j]
 \enspace .
\end{aligned}\end{equation}
By assumption, $\Pr[X = 0 \vert A = 0] = \alpha_{00} + \beta_{00} = \alpha_{01} + \beta_{01} = 1$.  

\begin{claim}
Any non-signaling strategy for the $\mathrm{CHSH}_n$ game with $\Pr[X = 0 \vert A = 0] = 1$ is a convex combination of classical strategies, i.e., strategies with $\Pr[X = 0 \vert A = a], \Pr[Y = 0 \vert B = b] \in \{0,1\}$ for all~$a, b$.  
\end{claim}

\begin{proof}
The constraints are the same as those for a unit flow through a directed graph with $4 n$ vertices, illustrated below for $n = 3$: 
\begin{equation*}
\includegraphics[scale=1.1]{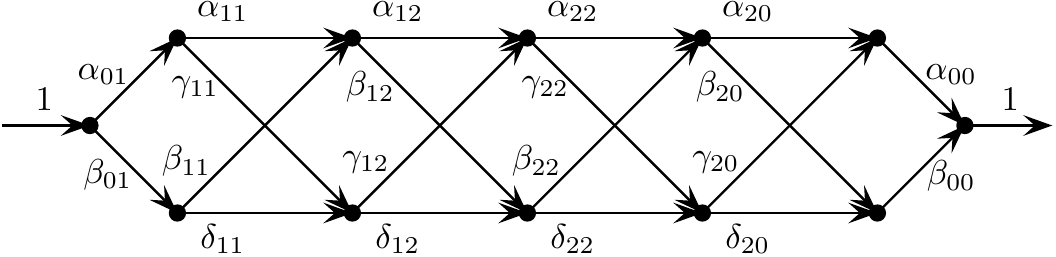}
\end{equation*}
Unit flows through the graph therefore correspond to strategies.  A unit flow is a convex combination of integer-valued flows, corresponding to a convex combination of classical strategies.  
\end{proof}

\noindent
Therefore, the non-signaling value with $\Pr[X = 0 \vert A = 0] = 1$ equals the classical value.  
\end{proof}

\subsection{Extended CHSH$_n$ game} \label{s:extendedchshngame}

For integers $k \geq 0$ and $n \geq 2$, we define the $k$-Extended CHSH$_n$ game, or ``$\mathrm{CHSH}_n + k$" for short, with players \Alice, \Bob\ and $\text{\Charliej{1}}, \ldots, \text{\Charliej{k}}$.  With $A$ and $X$ we denote \Alice's input and output, respectively, and similarly $B$ and~$Y$ for \Bob, and $C_j$ and $Z_j$ for \Charliej{j}.  The message ranges are $A \in \{-n+1, -n+2, \ldots, n-1\}$, $B \in \{0, 1, \ldots, n-1\}$ and $X, Y, C_j, Z_j \in \{0,1\}$.  

There are two types of questions: 
\begin{itemize}
\item 
In a \emph{consistency} question, the inputs are $A = 0$ and $C_J = 0$ for a uniformly random index $J \in [k]$.  The verifier accepts if $X = Z_J$.  
\item 
In a \emph{game} question, the verifier either sets $A = 0$, or chooses $A \in \{-n+1, \ldots, n-1\} \smallsetminus \{0\}$ uniformly at random.  The verifier chooses $B \in \{ \abs{A}, \abs{A} + 1 \pmod n \}$ uniformly at random, and if $A \neq 0$ sets $C_1 = \cdots = C_k = 1$.  Let $S = \chi_{A < 0}$ and $Z = Z_1 \oplus \cdots \oplus Z_k$.  The verifier accepts if 
\begin{gather*}
(A = 0 \text{ and } X = Y) \\
\text{or } \\
(A \neq 0 \text{ and } Z \neq S) \\
\text{or } \\
\Bigg( A \neq 0 \text{ and } Z = S \text{ and } \begin{cases}X \neq Y & \text{if $\abs{A} = B = n-1$} \\ X = Y & \text{otherwise} \end{cases} \Bigg)
 \enspace .
\end{gather*}
\end{itemize}
The verifier, conditioned on asking a game question, sets $A = 0$ with probability $p = \tfrac{1}{2n-1}$; this ensures that in game questions, $\Pr[A = 0 \vert A \geq 0] = \Pr[A = 0 \vert A \leq 0] = 1/n$.  The verifier asks a consistency question with probability $q = 1 / \big( 1 + \tfrac{1}{(3^k - 1) p} \big)$; this is $\Theta(1/n)$ for constant~$k$, and implies that conditioned on $A = 0$, game and consistency questions are comparably likely.  (The precise value is set to optimize \claimref{t:markov}.)  

Note that on input $A = 0$, \Alice\ cannot distinguish between a game and a consistency question, nor between the different consistency questions.  
Also note that the $\mathrm{CHSH}_2 + k$ game is different from $\mathrm{CHSH} + k$ defined in \secref{s:extendedchshgame}.  This game is more complicated because for $n > 2$ the verifier is generally unable to adjust \Alice's answer on input $A \neq 0$ to account for a possible $\sigma^z$ correction.  (This gives up a factor of two in the analysis.)  

\begin{proposition} \label{t:extendedchshnquantumvalue}
There exists a quantum strategy for the $\mathrm{CHSH}_n + k$ game, using the shared state $\tfrac{1}{\sqrt 2}(\ket{0^{k+2}} + \ket{1^{k+2}})$, such that 
\begin{align*}
\Pr[\mathrm{win} \vert \text{\emph{game question with $A = 0$}}] &= \valuequantum(\mathrm{CHSH}_n) \\
\Pr[\mathrm{win} \vert A > 0] = \Pr[\mathrm{win} \vert A < 0] &= \tfrac12 \big(1 + \valuequantum(\mathrm{CHSH}_n)\big) \\
\Pr[\mathrm{win} \vert \text{\emph{consistency question}}] &= 1
 \enspace .
\end{align*}
In particular, as $\valuequantum(\mathrm{CHSH}_n) = \cos^2 \tfrac{\pi}{4 n}$, 
\begin{equation*}
\Pr[\mathrm{win}] 
= 1 - (1 - q) \frac{n}{2n-1} \sin^2 \frac{\pi}{4 n}
> 1 - \frac{1}{2 n^2}
 \enspace .
\end{equation*}
\end{proposition}

\begin{proof}
We construct a quantum strategy.  

If her input $A < 0$, then \Alice\ applies $\sigma^z$.  She proceeds to play according to the optimal CHSH$_n$ strategy from \secref{s:chshn} and \figref{f:chshn}, for input $\abs A$.  In particular, if $A = 0$ then she measures~$\sigma^z$.  

\Bob\ plays according to the optimal CHSH$_n$ strategy, above.  

\Charliej{j}\ measures his qubit in the $\sigma^z$ basis on input~$0$, and in the $\sigma^x$ basis on input~$1$.  Therefore in a consistency question, \Alice\ and all \Charlie s get the same result, and the verifier accepts always.  

In a game question, then after the \Charlie s' $\sigma^x$ measurements, \Alice\ and \Bob\ share $\tfrac{1}{\sqrt 2}(\ket{00} + \ket{11})$, if $Z = 0$, or $(\sigma^z \otimes I) \tfrac{1}{\sqrt 2}(\ket{00} + \ket{11})$, if $Z = 1$.  
If $A = 0$, then they win with probability $\cos^2 \tfrac{\pi}{4n}$.  
If $A \neq 0$ and $S = Z$, then the shared state after \Alice's possible $\sigma^z$ correction is $\tfrac{1}{\sqrt 2}(\ket{00} + \ket{11})$, so they win with probability $\cos^2 \tfrac{\pi}{4n}$.  Since $\Pr[S = Z] = \tfrac12$, $\Pr[\text{win} \vert J = \mathrm{\Bob}, A \neq 0] = \tfrac12 (1 + \cos^2 \tfrac{\pi}{4n})$.  
\end{proof}

\end{document}